\documentclass[10pt,a4paper]{article}
\usepackage[utf8]{inputenc}
\usepackage[english]{babel}
\usepackage[T1]{fontenc}
\usepackage[dvipsnames]{xcolor}
\usepackage[unicode=true, colorlinks = true, allcolors=MidnightBlue]{hyperref}
\usepackage{mathtools}
\usepackage{amsfonts}
\usepackage{amssymb}
\usepackage{amsthm}
\usepackage{mathrsfs}

\usepackage{natbib}
\usepackage{float}

\usepackage[inline]{enumitem}
\usepackage{graphicx}
\usepackage[margin=1in]{geometry}
\usepackage{marginnote}

\usepackage{nicefrac}
\usepackage{setspace}
\usepackage{thmtools}
\usepackage{soul}
\usepackage{cancel}
\usepackage{shellesc}
\usepackage{pgfplots}
\pgfplotsset{compat=1.14}
\usepackage[noabbrev]{cleveref}

\theoremstyle{definition}
\newtheorem{theorem}{Theorem}[section]
\newtheorem{propo}[theorem]{Proposition}
\newtheorem{lemma}[theorem]{Lemma}
\newtheorem{coro}[theorem]{Corollary}
\newtheorem{defi}[theorem]{Definition}
\newtheorem{ex}[theorem]{Example}

\theoremstyle{remark}
\newtheorem{remark}[theorem]{Remark}

\newcommand{\defin}[1]{\textbf{#1}}

\newcommand{\B}{\ensuremath{B}}

\newcommand{\f}{\ensuremath{{\mathcal{D}_A}}}
\newcommand{\cog}{\mathcal{W}}

\newcommand{\ra}{\ensuremath{\mathbb{R}_+ \cup \{\infty \}}}
\newcommand{\A}{A}
\newcommand{\R}{\ensuremath{\mathbb{R}}}
\newcommand{\N}{\ensuremath{\mathbb{N}}}

\newcommand{\E}{\ensuremath{\mathbb{E}}}
\newcommand{\Pro}{\ensuremath{\mathbb{P}}}
\newcommand{\Li}{\ensuremath{L^\infty}}
\newcommand{\Lq}{\ensuremath{L^q}}
\newcommand{\Lp}{\ensuremath{L^p}}

\newcommand{\D}{\ensuremath{\mathcal{D}}}

\newcommand{\X}{\ensuremath{\mathscr{X}}}
\newcommand{\bd}{\operatorname{bd}}
\newcommand{\interior}{\operatorname{int}}
\newcommand{\closure}{\operatorname{cl}}
\newcommand{\conv}{\mathrm{conv}}
\newcommand{\clconv}{\text{cl-conv}}
\newcommand{\cone}{\operatorname{cone}}
\newcommand{\clcone}{\text{cl-cone}}
\newcommand{\dd}{\mathrm{d}}
\newcommand{\LR}{\mathrm{LR}}

 \newcommand{\Acc}[2]{\mathcal{A}^{#1}_{#2}}

\newcommand{\sta}{\operatorname{st}}
\newcommand{\pl}{{Minkowski Deviation}}
\newcommand{\pls}{{Minkowski Deviations}}

\DeclareMathOperator*{\esssup}{ess\,sup}
\DeclareMathOperator*{\essinf}{ess\,inf}

\title{Minkowski deviation measures}
\author{Marlon Moresco \\ email: {marlonmoresco@hotmail.com}
\and Marcelo Brutti Righi \\ email:{marcelo.righi@ufrgs.br}
\and Eduardo Horta \\ email: {eduardo.horta@ufrgs.br}}

\begin{document}
\maketitle

%\listoftheorems

\begin{abstract}

We propose to derive deviation measures through the Minkowski gauge of a given set of acceptable positions. We show that, given a suitable acceptance set, any positive homogeneous deviation measure can be accommodated in our framework. In doing so, we provide a new interpretation for such measures, namely, that they quantify how much one must shrink  or deleverage a position for it to become acceptable. In particular, the \pl\ of a set which is convex, stable under scalar addition, and radially bounded at non-constants, is a generalized deviation measure. Furthermore, we explore the relations existing between mathematical and financial properties attributable to an acceptance set, and the corresponding properties of the induced measure. Hence, we fill the gap that is the lack of an acceptance set for deviation measures. Dual characterizations in terms of polar sets and support functionals are provided.

\noindent \textbf{Keywords}: Risk measures, Deviation measures, Acceptance sets, Convex analysis, Minkowski gauges, \pls.
\end{abstract}

\section{Introduction}

In modern financial theory --- since the iconic paper of \citet{markowitz52} --- the \emph{standard deviation} has been the measure most used to quantify the risk of a financial position, especially in the framework of portfolio selection. More recently, due to the increasing necessity of paying attention to tail risks, monetary risk measures, which respect monotonicity and cash additivity, came to light. Following the seminal paper of \citet{artzner99}, theoretical properties that are desirable for a risk measure have been widely studied, but no consensus has been reached so far about which set of axioms are the most adequate (in terms of generality, applicability, theoretical tractability, etc.). The axiomatic approach of \citet{rockafellar06} represents a landmark in the literature, setting the tone for recent developments with the introduction of \emph{generalized deviation measures} --- generalizations of the standard deviation and similar measures which capture the ``degree of non-constancy'', or dispersion, of a financial position. Such measures have been proved useful in financial problems as can be seen in \citet{rockafellar06b}, \citet{pflug06}, \citet{grechuk09}, \citet{rockafellar13} among others. In this context, and due in part to the aforementioned lack of a universal approach to measure risk, a handful of coherent and convex \emph{risk} measures have been proposed and, as a dénouement, many generalized and convex \emph{deviation} measures as well. In addition, \citet{righi16}, \citet{berkhouch18} and \citet{righi19} bring forward some novel convex risk measures, in the sense of \citet{follmer02}, which explicitly take variability into account. Empirically, {this class of convex, ``dispersion aware'' risk measures has been shown to display a consistently better performance for optimal portfolio strategies, as seen in the work of} \citet{righi17b}.

\medskip
\textsc{In the present paper, we bring forward a novel way to obtain deviation measures}. Drawing inspiration from the canonical representation of a monetary risk measure as an infimum over the set of acceptable cash additions on a given position, we propose using the well-known \emph{Minkowski gauge} from Functional Analysis as a means to recover, from a given admissible set of acceptable positions, an implicit deviation functional.
% which we shall call the \emph{\pl}\ corresponding to that given set of acceptable positions.
Our approach indicates that, from a financial perspective, a numerically quantified measure of risk/deviation may be seen as a derived concept: one can always take acceptance sets as the fundamental building blocks. We show that if the requirement is met, that sensibility to expanding/shrinking a position is homogeneous with respect to the scale of expansion/shrinkage, then each admissible set of acceptable positions gives rise to a deviation functional, which we shall refer to as the \pl\ implied by said acceptance set. An important result which we prove herein is that \emph{\pls\ exhausts the class of positive homogeneous deviation measures}. This proposition suggests a novel way to interpret certain deviation measures --- which are commonly seen as functionals that quantify the \emph{distance} between a random variable and constancy --- as functionals that capture the amount that an agent must shrink a given position for it to be considered acceptable.

Formally, ours \pl\ is a functional defined on a space $\X$ comprised of a suitable class of random variables which represent feasible financial outcomes. The generic element $X\in\X$ is understood as a real-valued, random result of a financial asset, corresponding to a certain position whose realized value depends on the outcome $\omega$ of the market, and we adopt the convention that $X(\omega)>0$ denotes a gain. It is important to highlight the generality of our framework, in that we impose little {structure} on the space $\X$---only requiring that it be a topological vector space---, thus encompassing the most used spaces in the literature, such as the $L^p$ and Orlicz spaces. Although it is possible, \emph{in principle}, to interpret an arbitrary functional $f\colon\X\to \R\cup\{+\infty\}$ as representing the \emph{financial risk} of a position $X$ (through the value $f(X)$), it is customary in the literature to restrict attention to two broad classes of functionals, namely the class of \emph{monetary risk measures} and the class of \emph{deviation measures}.\footnote{The tenured reader is probably familiar with the fact that the terminology \emph{monetary risk} and \emph{deviation} ``measure'' is misleading as the objects under study are not bona fide measures (as in ``$\sigma$-finite measure'' for instance) but rather \emph{functionals} (possibly non-linear) on a topological vector space.} \pls\ fall in the second category, and --- as mentioned above --- coincide with the class of positive homogeneous deviation measures, in the sense that any such measure can be represented in the form
\begin{equation}\label{minkowski-functional-def}
\f(X) = \inf\{m>0\colon\,m^{-1}X\in A\},\quad X\in\X,
\end{equation}
for a suitable $A\subseteq\X$. This \emph{representation theorem} is one of the central messages of this paper, standing in analogy to the aforementioned representation theorem according to which any monetary risk measure can be expressed canonically in the form
\begin{equation}\label{eq:monetary-risk-measure-via-infimum}
\rho_A(X) = \inf\{m\in\R\colon\,X+m\in A\},\quad X\in\X,
\end{equation}
for a suitable $A\subseteq\X$. In other words, whereas monetary risk measures are representable as the \emph{minimum translation factor} (corresponding to cash addition/subtraction) which makes a given position acceptable, for positive homogeneous deviation measures the proper concept is that of a \emph{least scaling factor} (corresponding to expansion/shrinkage) which makes said position acceptable, and the function which captures the latter idea is precisely the \emph{Minkowski gauge} in \cref{minkowski-functional-def}. See \Cref{fig:minkowski-gauge}.

The \pl\ in \cref{minkowski-functional-def} has an underlying acceptance set $A$ which can be quite arbitrary, at least in principle. In practice, it must be ``sufficiently rich'' in order that the yielded \pl\ be of interest: we show, for example, that under some weak assumptions on $\f$ it is always the case that $A$ in \cref{minkowski-functional-def} is of the form $A = \{X\in\X\colon\,\f (X)\le1\}$. Again there comes to light a similarity to the typical representation of the underlying acceptance set of a monetary \emph{risk} measure: as mentioned above, an arbitrary such functional, say $\rho$, is by necessity of the form given in \cref{eq:monetary-risk-measure-via-infimum}, with $A = \left\{X\in\X \colon\, \rho(X) \leq 0 \right\}$. Of course, if $\rho$ were instead a \emph{deviation} measure, then the latter $A$ would deem only constants as acceptable. By the same token, it is clear that, in general, a deviation measure $D$ is not representable in the form $D(X) = \inf\left\{m\in\R \colon\, X + m \in A \right\}$ for some $A\subseteq\X$. In summary we have the following scheme of implications: on the one hand, there is the classical result which states that --- under suitable assumptions on the set $A$ --- the functional $\rho$ defined by \cref{eq:monetary-risk-measure-via-infimum} is a monetary risk measure, and, reciprocally, if $\rho$ is a monetary risk measure, then it can be written as in \cref{eq:monetary-risk-measure-via-infimum} with $A = \left\{X\in\X \colon\, \rho(X) \leq 0 \right\}$. On the other hand, and this is one of the main contributions of the present paper, we show that --- again under suitable assumptions on $A$ --- the functional $\f$ defined by \cref{minkowski-functional-def} is a positive homogeneous deviation measure, and, reciprocally, if $D$ is any positive homogeneous deviation measure, then $D = \f$ with $\f$ given in \cref{minkowski-functional-def} and $A = \{X\in\X\colon\,\f (X)\le1\}$. This shows, in particular, that the notion of a \emph{set of acceptable positions} must be distinct whether one has in mind monetary risk measures or, instead, deviation measures: for the latter, the ``correct'' approach is to consider a position acceptable (with regards to a deviation measure $D$) if it lies in the sub-level set $\{X\in\X\colon\, D(X)\le1\}$ or, more generally, in a sub-level set
\begin{equation}\label{eq:sub-level-sets}
\Acc{k}{D}\coloneqq \{X\in\X\colon\, D(X)\le k\},
\end{equation}
where $k>0$ is some prescribed constant.

At the heart of our approach, notwithstanding, is the message that one can take acceptance sets as the ``datum of the problem''. In other words, we argue that financially it makes sense to pass \emph{from the set to the measure} in contrast to the purely algebraic passage from the measure to the acceptance set.
 It only turns out that, quite conveniently, any ``admissible'' acceptance set $A$ is by necessity  ``nearly'' of the form $A = \Acc{1}{\f}$ where the precise meaning of ``nearly'' is given in  \Cref{lemma-item4}.
 %among the ``suitable assumptions on $A$'' that we mentioned above, we highlight convexity (a financially motivated arequirement, but this can be weakened), under which
 In this milieu, one possibility could be to adapt the approaches put forth by \citet{frittelli06} and \citet{artzner09}. These can be outlined as follows: there are multiple eligible assets whose aim is to recover, from a given set $A$, an implicit measure through $ \rho_A (X) = \inf \left\{\pi (Y)\colon X + Y \in A \right\} $, where $\pi\colon \mathcal{C} \rightarrow \R$ is the cost to execute $Y$, and $\mathcal{C}$ is a set of feasible strategies. However, the preceding infimum yields a measure which is neither translation insensitive nor non-negative --- not a problem if one has risk measures in mind, but an impassable hurdle if the aim is to obtain measures of \emph{deviation}. An alternative within reach is to assume that there exists some (constant) risk-free asset $c$, in which case --- for a given position $X$ and an acceptance set $A$ --- we can use convexity to reduce the position's risk, up to the point where it becomes acceptable; in other words, by recovering the measure implied by $A$ via $\mathfrak{D}_A (X) = \inf \{\lambda \in [0,1) \colon\, (1-\lambda) X + \lambda \pi (X)c \in A \} $ this is an intrinsic risk measure as developed by \cite{farkas19}. 
Their intrinsic risk measure is  the smallest percentage of the currently held financial position
which has to be sold and reinvested in an eligible asset such that the resulting
position becomes acceptable.
%  In the special case when $A$ is the sub-level set $\Acc{1}{D}$ corresponding to some previously given convex deviation measure $D$ (which entails convexity of $A$), we have the equivalences
%\[
%\mathfrak{D}_{A} (X)
%%&= \inf \{\lambda \in [0,1) \colon\, (1-\lambda) X + \lambda c \in \Acc{1}{D} \} \\
%= \inf \{\lambda \in [0,1) \colon\, D\big((1-\lambda) X + \lambda \pi (X) c \big) \leq 1 \} 
%= \inf \{\lambda \in [0,1) \colon\, D\big((1-\lambda) X\big)\leq 1 \} . 
%\]
%Thus, the quantity $\mathfrak{D}_{A}(X)$ can be understood as the amount by which we must shrink the position $X$ until it becomes acceptable (with regards to $D$). 
There is an important drawback in this approach, however --- namely, that any two acceptable positions will always have the same measurement, whereas in general we wish to be able to distinguish the ``better'' position. Furthermore, this intrinsic risk measure is not convex. Additionally, there exist no practical measure which can be classified as an intrinsic risk measure. Our approach has the same intuition, but without its drawbacks, in fact, any positive homogeneous deviation measure will be covered in our approach. What is more, our work gives the powerful interpretation of shifting a position to acceptability to deviation measures.

The above discussion reiterates the fact that, from a financial perspective, the idea of shrinking and expanding a position is closely related to the concept of positive homogeneity, more so if we interpret the numerical quantification of risk/deviation as merely an echo stemming from an underlying operation taken on the fundamental acceptance set. Indeed, for a positive homogeneous deviation measure $D$, we can interpret the mapping $\lambda\mapsto D(\lambda X)$, where $\lambda>0$, as controlling simultaneously the size and the deviation of the position $X$. It appears only natural, then, to stipulate that a measure of `non-constancy' is positive homogeneous. This requirement is reinforced by the consideration that most of the prominent deviation measures found in the literature are indeed positive homogeneous --- besides, many relevant deviation measures that are not so, such as the variance and the entropic deviation, are only one transformation away from positive homogeneity (for instance, the standard variation in relation to the variance, etc.). See \citet{follmer11} for details on the positive homogeneous approximation of the entropic deviation. In summary, positive homogeneity of $D$ should translate into the following two properties for the corresponding acceptance set: in case the position $X$ does not lie in $\Acc{k}{D}$, we should be able to shrink the position until it ``fits'' in the set. Reciprocally, if $D(X)\le k$, then we should be able to enlarge the position up to a limit where it still lies in the set. This is exactly the idea that the \pl\ in \cref{minkowski-functional-def} describes. Additionally, under positive homogeneity, acceptance sets of the form $\Acc{k}{D}$ generated by a deviation measure $D$ at a certain level $k$ admit a compelling financial interpretation: namely, that $k$ represents an agent's coefficient of aversion with respect to $D$. Also, $k$ can be chosen to be some benchmark level, say $k = D(I)$ where $I\in\X$ is a relevant index. Obviously, an agent with greater $k$ has higher compliance regarding exposure to dispersion, so that, in order to compare positions of agents with varying degrees of aversion, we must bring the deviation measure to the same level for all market participants. This is so, even if the distinct agents agree about which deviation measure should be used, in which case positive homogeneity allows us to normalize each set of the form $\Acc{k}{D}$ by the factor $1/k$, yielding the identity $\f(X) =k\cdot \inf \left\{m >0 \colon\,{D(X)} \le mk \right\} $ with $A = \Acc{1}{D}$. Last but not least, it is reasonable to assume (and we do so throughout the text) that it is possible to invest the excess capital resulting from shrinkage (similarly, to borrow the demanding capital for the enlargement) into a constant risk-free asset, i.e., to require that acceptance sets be stable with respect to translation by a constant. In other words, adding a constant to a given position has no effect on whether the latter is acceptable or not. This property is true, in particular, whenever $A$ is generated by a deviation measure (i.e., $A = \Acc{k}{D}$), in which case, owing to translation insensitivity, allocation of capital in a risk-free manner leads to no change in the deviation of the position.

The idea of studying deviation measures through the lens of Minkowski gauges is not entirely new. \cite{pflug07} previously explored this terrain. However, the authors restrict attention to functionals $\f$ implied by sets of the form $A = \{X\in\X \colon\, \mathbb{E}(h \circ X)\leq h(1)\} $ for a convex, symmetric, non-negative real function $h$ with $h(0) = 0$ and $0<h(x)<\infty$ for $x\neq0$, thus establishing a relation between financial risk and Orlicz norms. In particular, if $h$ is invertible on $[0,+\infty)$, then the set $A$ is a sub-level set of the form $\Acc{1}{f}$, with the functional $f$ constrained to be of the form $f(X) = h^{-1}\big(\E(h\circ X)\big)$, in particular, they do not make the connection of deviation measures with acceptance sets, which illustrates once again that we are approaching the subject with greater generality. In any event, the authors propose deviations of the form $\f (X - \E X)$ and $\f ((X - \E X)^-)$, and explore to exhaustion the different representations of this kind of functional. A homologous approach was studied in \cite{bellini18}, who consider \emph{return risk measures} $\tilde{\rho}$, which are analogous to monetary risk measures but applied to the return of a position, not its profit/loss. A return risk measure is a functional $\tilde{\rho}$ defined on the cone of strictly positive returns $\{X \in \Li(\Omega, \mathfrak{F},\Pro) \colon\, X>0\}$ which maps into the half line of strictly positive real numbers. Such an $\tilde{\rho}$ is also positive homogeneous, satisfies $\tilde{\rho} (1) = 1$, and stays in a one-to-one correspondence with a monetary risk measure $\rho$ via the relation $\tilde{\rho}(X) = \exp (\rho (\log (X)))$. Indeed, given a suitable acceptance set $A = \Acc{1}{\tilde{\rho}}$ the return risk measure can be precisely recovered through the \pl\ of $A$, i.e., $\tilde{\rho} = \D_{\Acc{1}{\tilde{\rho}}}$.

The remainder of this paper is structured as follows: section \ref{sec preli} introduces our notation and framework, and also provides the underlying financial intuition backing set and functional properties that shall be used throughout this paper. In Section \ref{sec:deviations} we explore the \pl\ as a deviation measure, developing the role of specific properties for the set and its impact on the properties for the implied functional. In section \ref{accept} we develop the idea of an acceptance set generated by a deviation measure by exploring the reverse implications from section \ref{sec:deviations}. The \cref{apen A} contains some results regarding Minkowski gauge as an abstract functional and some auxiliary results. \Cref{apen figs} houses some figures to help in developing the intuition behind the set properties and or functional.

\section{Preliminaries and some set properties}\label{sec preli}
The notion of an \emph{acceptance set} is a cornerstone in defining our \pl, the idea being that such set determines the ``range'' of financial positions whose risk is deemed acceptable. \citet{artzner99} were the first to propose the concept, after which it was deepened, among others, by \citet{delbaen02}, \citet{frittelli06}, and \citet{artzner09}. In this section we wrap up the necessary terminology which, although not entirely new, is somewhat scattered throughout the literature. We also provide some compelling financial interpretation behind many concepts familiar to the convex analyst, showing that purely mathematical properties (for example, star-shapedness) can be given an intuitive meaning when seen as attributes of a given acceptance set. The reader may skip straight to \cref{sec:deviations} if she is too eager to see some action, and come back here for the definitions as needed.

In all that follows, $(\Omega,\mathfrak{F}, \Pro)$ is a fixed probability space. Every equality and inequality involving random variables is to be understood as holding $\Pro$-almost surely.\footnote{Some care is needed, however, in the face of a relation of the type ``$X\in B$'', as this could mean that $\Pro(X\in B) = 1$ for $B\subseteq\R$ but have a completely different meaning when $B\subseteq\X$.} As usual, we write, for $p\in(0,\infty)$, $L^p \equiv L^p(\Omega,\mathfrak{F},\Pro) \coloneqq$ ``the set of all ($\Pro$-equivalence classes of) random variables $X$ such that $\E|X|^p < \infty$'', whereas $L^0\equiv L^0(\Omega,\mathfrak{F},\Pro) \coloneqq$ ``the set of all ($\Pro$-equivalence classes) of random variables on $(\Omega,\mathfrak{F},\Pro)$'', and $L^\infty \equiv L^\infty(\Omega,\mathfrak{F},\Pro)\coloneqq$ ``the set of all ($\Pro$-equivalence classes of) random variables $X$ which are $\Pro$-essentially bounded''.
We work with a Hausdorff topological vector space $\X$, and assume beforehand that the inclusions $L^0\supseteq\X\supseteq L^\infty$ hold.\footnote{{These inclusions are assumed to hold algebraically --- no \emph{a priori} assumption is made on the relation between the topologies involved.}} The generic elements of $\X$ are denoted by $X$, $Y$, $Z$, etc., and are to be interpreted as the random result of a \emph{financial position}, which we assume throughout to be perfectly liquid and discounted by a risk-free rate. $\X' $ denotes the topological dual of $\X$, and we shall write $\langle X,X'\rangle \coloneqq X'(X)$ whenever $X\in\X$ and $X'\in\X' $; notice that this notation gives $\langle X, Y\rangle = \E XY$ if $X\in L^p$ and $Y\in L^q$, with $1\leq p<\infty$ and $p^{-1}+q^{-1} = 1$, via the identification $L^q\equiv (L^p)'$. {Furthermore, we write $\langle \X , \X' \rangle = \X\times\X'$, and call this construct the \defin{dual pair}. With this notation and terminology, the mapping $(X,X')\mapsto\langle X, X'\rangle$ gives a bilinear functional defined on the dual pair, one that separates points of both $\X$ and $\X'$}. The positive and negative parts of an element $X\in\X$ are denoted by $X^+ \coloneqq \max(X,0)$ and $X^- \coloneqq \min(-X,0)$, respectively. We define the cone $\X_+$ of non-negative positions as $\X_+\coloneqq \{X\in\X\colon\, X\geq0\}$ (this is the range of $X\mapsto X^+$), and similarly $\X_- \coloneqq \{X\in\X\colon\,X\leq0\}$.
With a slight abuse of notation, we consider the inclusion $\R\subseteq \X$ by identifying each $x\in\R$ with the equivalence class of random variables equal to $x$ almost surely. A pair of random variables is said to be \defin{comonotone} if the inequality
\[
(X(\omega) - X(\omega')) (Y(\omega)-Y(\omega') ) \geq 0,\qquad \omega,\omega'\in \Omega
\]
holds $\Pro\otimes\Pro$-almost surely. As usual, $F_X$ represents the cumulative distribution function of a random variable $X$, while $F_X^{-1}$ denotes its left quantile function, that is to say, $F_X^{-1}(\alpha) \coloneqq \inf\{q\in\R\colon\, F_X(q)\geq\alpha\}$. We write $X =_d Y$ whenever $X$ and $Y$ are equal in distribution, a fact which we also express by writing $Y\in\mathcal{L}_X$ (and this already defines $\mathcal{L}_X$ implicitly). As mentioned, we denote the property of $X$ being almost surely greater than $Y$ by $X \geq Y$, while for a generic partial order $\succeq$ we write $ X \succeq Y$, also adopting the obvious convention that the notation $ X \preceq Y$ means precisely that $ X \succeq Y$. If not clear from context, we shall mention explicitly the partial order under consideration. We say that \defin{$X$ is greater than $Y$ in the dispersive order of distributions}, written $Y \preceq_{\mathfrak{D}} X$, if the inequality $F_X^{-1} (u) - F_X^{-1} (v) \geq F_Y^{-1}(u) - F_Y^{-1}(v)$ holds for every $0 < v < u < 1$. In all that follows, $\R_+$ denotes the set $[0,+\infty)$, whereas $\R_+^* \coloneqq (0,+\infty)$. 

{Given $A,B\subseteq\X$ we define the set $ \A + \B$ by saying that $Z\in A+B$ if and only if $Z = X+Y$ for some $X\in A$ and some $Y\in B$. Similarly, for a $\Lambda\subseteq\R$, we write $Z\in\Lambda \A$ if and only if $Z =\lambda X$ for some $\lambda\in\Lambda$ and some $X\in A$. For simplicity, we write $\lambda A \coloneqq \{\lambda\} A$ and $\Lambda X \coloneqq \Lambda \{X\}$ when one of the involved sets is a singleton; in particular, we define the ray of $X\in\X$ as $R_X \coloneqq \R_+^* X$. In the same manner, $X + A\coloneqq \{X\} + A$, etc.
We also denote by $\bd(A)$, $\interior(A)$, $\closure(A)$, $\conv(A)$, $\clconv(A)$, $\cone(A)$, $\clcone(A)$, and $A^\complement$ respectively the boundary, interior, closure, convex hull, closed convex hull, conic hull, closed conic hull and the complement of $A$. Any $A\subseteq\X$ is called an \defin{acceptance set}, and we say that a given position $X$ is \defin{acceptable} (w.r.t.\ $A$) if and only if is an element of $A$.}

We now focus on properties for sets that are considered alongside the text. As said above, we make an effort to clarify the financial intuition behind each of these attributes. Since not every property appearing in our axiom scheme is fundamental in functional and convex analysis --- and thus it is likely that some of these attributes are unknown to the reader ---, we shall resort to figures as a means to illustrate them and help to develop the intuition. In these figures, we are considering $\Omega$ as the \emph{binary market}, i.e., $\Omega = \{0,1\} $; in this setting, one can take $\X = L^0 \equiv \R^2$, where the latter equivalence is given via the identification of a random variable $X$ with the ordered pair $\big(X(0), X(1)\big)$ in the Cartesian plane. Importantly, notice that in this context the inclusion $\R\subseteq \R^2$ corresponds to the diagonal $\{(u,v)\colon\, v=u,\,u\in\R\}$, which may be different from what the reader has in mind at first thought.

\begin{defi}
Let $A\subseteq\X$ and $\{A(k) \colon\,k\in\R\} \subseteq 2^\X $. We say that
\begin{enumerate}[label = (\roman*)]

\item ({\sc Law invariance}) $\A$ is \defin{law invariant} if $X \in \A$ and $X =_d Y$ implies $Y \in \A$.

This means that a financial position having the same distribution as a given, acceptable position is also acceptable; {that is, when deciding whether a position is to be deemed acceptable, we only care about its statistical properties}.

\item ({\sc Monotonicity}) $A$ is \defin{monotone} with respect to a given partial order~$\preceq$ if the conditions $X \in \A$ and $ X \preceq Y $ imply $Y \in \A$. $A$ is said to be \defin{anti-monotone} (w.r.t~$\preceq$) if the conditions $Y \in \A$ and $ X \preceq Y $ imply $X \in \A$. For convenience, we say that $A$ is \defin{$\preceq$-monotone} whenever $A$ is monotone with respect to $\preceq$, and similarly for anti-monotonicity.

{Under monotonicity, a position is deemed acceptable whenever a ``worse'' (smaller) one is also acceptable (from a financial perspective, this is not very interesting). Anti-monotonicity, on the other hand, captures the notion that being ``bigger'' according to some partial order is actually worse, e.g.,\ the dispersive order of distribution. Under anti-monotonicity, then, a position is regarded as acceptable whenever a ``better'' position is also acceptable. Note that if $A$ is monotone then $A^\complement$ is anti-monotone: indeed, letting $A$ be monotone and $X \preceq Y$, then $X \in A$ implies that $Y \in A$, which is equivalent to say that $ Y \notin A$ implies that $X \notin A$, thus yielding anti-monotonicity of $A^\complement$.}

% \item ({\sc Conicity}) $A$ is a \defin{cone} if $\lambda X \in A$,
% for all $X \in A$ and $\lambda \in \R_+$.
% $A$ is a \defin{cone with vertex at $V\in\X$} if $V+\lambda X \in A$, for all $X \in A$ and $\lambda \in \R_+$. Whenever we say simply that $A$ is a \defin{cone} we are implicitly assuming that $A$ is a cone with vertex at $0$.
\item ({\sc Conicity}) $A$ is a \defin{cone with vertex at the origin}, or simply a \defin{cone}, if $\lambda X\in A$ for every $\lambda\ge0$ and every $X\in A$. $A$ is said to be a \defin{cone with vertex at $V\in\X$} if $A$ is of the form $A = V+C$ for some cone $C$. A cone with vertex at $V$ is \defin{degenerate} if it is a singleton; otherwise, it is said to be a \defin{proper cone with vertex at $V$}.

Conicity means that if a position is acceptable, then every non-negative multiple of the position is deemed acceptable as well. This is a reasonable assumption when we are concerned with losses, but not so much for dispersion, as it allows scaling any acceptable position \emph{up} in an unbounded fashion.

\item ({\sc Radial boundedness}) $ A $ is \defin{radially bounded} if, for every non-zero $X \in A$, there is some $ \delta_X \in (0 , \infty)$, such that $ \delta X \notin A$ whenever $\delta \in [\delta_X , \infty)$. The set $ A $ is said to be \defin{radially bounded at non-constants} if $ A \backslash \R$ is radially bounded.

Radial boundedness is, in a sense, the opposite of conicity: it says that there is always a bound on how much it is possible to scale up a position while keeping it acceptable. It means precisely that $A$ contains no cone (except for the trivial cone $\{0\} $) --- see \Cref{fig radially bounded} for an example. As constants have no dispersion, financially it makes sense to always consider them acceptable; that is to say, when we are mainly concerned with positions that are acceptable with respect to their dispersion, it is fruitful to limit the scaling up of all positions except for constants. In this case we should require that $A$ be radially bounded at non-constants. \Cref{fig stable scalar addition} shows a set which is radially bounded at non-constants but it is not radially bounded.

\item ({\sc Stability under scalar addition}) $ A $ is \defin{stable under scalar addition} if $A + \R = A$, that is, if $ X + c \in A$, for all $X \in A$ and $c \in \R$.

In our framework, as scalar addition does not affect the dispersion of a financial position, it is a {reasonable} property to be imposed on acceptance sets --- the set $A$ in \Cref{fig stable scalar addition} is stable under scalar addition, whereas the one in \Cref{fig radially bounded} is not. 

It is important to note that stability under scalar addition is incompatible (from a financial perspective) with monotonicity (or anti-monotonicity) with respect to some partial orders of interest, such as the ``almost surely $\geq$'' order.
To illustrate, assume $A$ is $\leq$-monotone, stable under scalar addition and that $0\in A$. Then $\Li \subseteq A$: indeed, since $0\in A$, stability under scalar addition immediately entails $\R \subseteq A$. Then, for any $Y \in \Li$ it follows that $Y \geq\essinf Y \in \R \subseteq A $, so monotonicity gives us $Y \in A$. Clearly, such an $A$ is way too large to be of any practical interest from a financial perspective. Also, stability under scalar addition is clearly incompatible with radial boundedness, as a non-empty acceptance set that respects stability under scalar addition contains at least the whole real line, and hence it cannot be radially bounded. However, a set which is radially bounded at non-constants, such as the one in \Cref{fig stable scalar addition}, undoubtedly can accommodate stability under scalar addition.

\item ({\sc Absorbency}) $ A $ is \defin{absorbing} if, for every $X \in \X$, there is some $\delta_X > 0$ such that $[0, \delta_X]X\subseteq A $, that is, such that $\lambda X\in A$ whenever $0\leq \lambda\leq \delta_X$.

$A$ being absorbing means that, for any random variable $X\in \X$ (not necessarily in $A$), the line segment joining $0$ to a suitable rescaling of $X$ lies entirely in $A$. Absorbing sets are of interest in part because any positive homogeneous function is completely determined by its values on any absorbing set. Furthermore, when $A$ is absorbing, it is possible to shrink any position until it ``fits'' in the set, and such that any further shrinkage of the position will keep it inside the set. In other words, any position may be scaled down to a point where it becomes acceptable. Importantly, in a topological vector space, every neighborhood of zero is an absorbing set.
\Cref{fig absorbing set} shows an example of an absorbing set.

\item ({\sc Convexity)} $ A $ is \defin{convex} if $ \lambda X + (1- \lambda Y) \in A$, for every pair $X, Y \in A $ and every $\lambda \in [0,1] $.

{Convexity is a fundamental property in the theory of vector spaces}. In our context, it is closely related to the concept of diversification, in the following sense: if an acceptance set $A$ is convex, then one cannot obtain an unacceptable position via a convex combination of acceptable positions, i.e.,\ we cannot get worse off when we diversify. Analogously, if the complement of an acceptance set $A$ is convex, then we cannot get better off by taking convex combinations of non-acceptable positions.

\item ({\sc Star-shapedness}) $ A $ is \defin{star-shaped} if $\lambda X\in A$, for every $X \in A$ and $\lambda \in [0,1]$. $ A $ is said to be \defin{costar-shaped} if $ A^\complement $ is star-shaped.

%\textcolor{orange}{Star-shapedness at $Z$ is a weaker form of convexity. When the convex combination of $Z$ and any other point in $A$ is in $ A$, note that is $A$ is star-shaped at $Z$, $A-Z$ is star-shaped at $0$ which is weaker then conicity.}
{$A$ being star-shaped means that the line segment joining $0$ to $X$ lies entirely in $A$, for every $X$ \emph{already lying in $A$} (thus, star-shapedness does not imply absorbency).} For a star-shaped set $A$, given any $X\in \X$, there exists some non-negative number $\lambda_X$ (possibly with $\lambda_X=\infty$) such that that $\R_+ X\cap A \supseteq (0,\lambda_X)X$ and $\R_+ X\cap A^\complement \supseteq (\lambda_X,\infty)X$; note that if $A$ is absorbing then we can take $\lambda_X >0$, and if $A$ is radially bounded then we can take $\lambda_X < \infty$. For sets containing zero, star-shapedness is a slightly weaker requirement than convexity: if $0\in A$ and $A$ is convex, then $A$ is star-shaped. \Cref{fig star-shaped} displays a star-shaped set which is not absorbing nor convex, while \Cref{fig absorbing set} shows a set that is not star-shaped, although absorbing. {Notice that $A\neq\varnothing$ being costar-shaped implies $\lambda X \in A$, for every $X \in A$ and $\lambda \in (1,\infty)$.}

We let $\sta(A)$ denote the \textbf{star-shaped hull} of $A$, which is defined by the condition that $Z\in\sta(A)$ if and only if $Z = \lambda X$ for some $\lambda\in[0,1]$ and some $X\in A$ (that is, $\sta(A) = [0,1]A$ in our preceding notation). It is clear that $\sta(A)$ is the smallest star-shaped set that contains $A$. Also, as an arbitrary intersection of star-shaped sets is still star-shaped, we see that $\sta(A)$ is equal to the intersection of all star-shaped sets that contain $A$. 

{Star-shapedness captures the financial notion that any scaled down version of an acceptable position should also be deemed acceptable}. This is clearly a desirable property, as it intuitively means that if an agent accepts to invest a certain amount in a stock, then she also finds it acceptable to invest a lesser amount in the same stock.

\item ({\sc Strong star-shapedness}) $A$ is \defin{strongly star-shaped} if $A$ is star-shaped and, for each $X\in\X$, the ray $R_X\equiv (0,\infty)X$ intersects the boundary of $A$ at most once, i.e.,\ the set $R_X \cap \bd\, A$ is either empty or a singleton. For a similar concept, see \cite{rubinov04}. \Cref{fig:strong-star} provides an example of a strongly star-shaped set having the origin as a boundary point. This is a technical concept.

%\item ({\sc Positive homogeneity}) $ \left\{A(k)\colon\, k \in \R\right\}$ is \defin{positive homogeneous} if it holds that $ \lambda A(k) = A(\lambda k)$, for each $\lambda > 0$.
%
%{A positive homogeneous family of sets $ \left\{A(k)\colon\, k \in \R\right\}$ may be seen as arising from a group of agents who take into account the same facts to quantify the risk, or dispersion, of a financial position}. In such a family, each agent would have a particular tolerance $k$ to dispersion; for example, if agent A is twice as tolerant as agent B (who accepts any position in the set $A(k)$, say), then the acceptance set of agent A would be $A(2k)$. {For simplicity, we introduced the concept here letting the family of sets be indexed by the whole real line, but the index set could be $\R_+$ or $\R_+^*$ as well without any modification in the definition.}

%\item ({\sc Symmetry}) $A$ is \defin{symmetric} if $X \in \A $ implies $-X \in \A$.
%
%While symmetry is useful, specially as --- whenever $\X$ is a normed space --- open balls centered at the origin are symmetric, this attribute is not desirable from a financial perspective, when $A$ represents a collection of acceptable positions. Indeed, there is no reason to require nor to expect that, for a given portfolio $X$ which is deemed acceptable, the corresponding short position $-X$ should be considered acceptable as well.

\end{enumerate}
\end{defi}

Before moving on to study the \pl in depth, we briefly turn our focus to relevant properties --- which regard functionals in general, not only the \pl\ --- that are considered alongside the text.

\begin{defi}
Let $f\colon\X \rightarrow \R \cup \{\infty \} $ be an arbitrary, extended real-valued functional on $\X$. A \defin{sub-level set} of a functional $f$ (defined on $\X$) at level $k\in\R$ is denoted by $\Acc{k}{f} \coloneqq \{X \in \X \colon\, f (X) \leq k \} $. Moreover, we say that 
\begin{enumerate}[label = (\roman*)]

\item\label{no_negativity} ({\sc Non-negativity}): $f$ is \defin{non-negative} if $f (X) > 0$ for any non-constant $X$ and $f(X) = 0 $ for any constant $X$.

If $f$ is a deviation measure, non-negativity tells us that that the deviation can only assume strictly positive values, except when evaluated at constants --- which have no deviation.

\item\label{trans_inva} ({\sc Translation insensitivity}) $f$ is \defin{translation insensitive} if $f (X + c) = f (X)$ for any $X\in\X$ and $c \in \R$.

Whenever $f$ is a deviation measure, translation insensitivity ensures that the deviation does not change if a constant amount is added to a given position.

\item ({\sc Monotonicity}) $f$ is \defin{monotone} (w.r.t.\ a given partial order $\preceq$) whenever $ Y \preceq X$ implies $f(Y) \leq f(X)$. If $-f$ is monotone, then $f$ is said to be \defin{anti-monotone} (w.r.t.\ $\preceq$). For simplicity, whenever the partial order is not explicitly mentioned, we are assuming that it is the ``almost surely $\leq$'' partial order.

From a financial perspective, imposing anti-monotonicity on a risk-functional $f$ corresponds to the requirement that, if a position yields better results than another in every possible \emph{state of the world}, then the former necessarily has lower risk than the latter.

\item\label{posi_homo} ({\sc Positive homogeneity}) $f$ is \defin{positive homogeneous} if $f(\lambda X) = \lambda f(X)$ for all $X\in\X$ and all $\lambda \geq 0$.

For a risk measure $f$, positive homogeneity has the financial interpretation that the risk of a position increases proportionally to its magnitude, capturing thus a type of homogeneous sensibility to expansion/shrinkage.

\item\label{convex} ({\sc Convexity}) $f$ is \defin{convex} if $f(\lambda X + (1- \lambda)Y) \leq \lambda f(X) + (1- \lambda ) f(Y)$, for every pair $X,Y\in\X$ and all $\lambda \in [0,1]$.

From the financial viewpoint, convexity is a property which ensures that diversification reduces risk. A mapping $f\colon\X\to\R\cup\{+\infty\}$ with $f(0)=0$ is said to be a \defin{sub-linear} functional whenever it satisfies any two\footnote{It is well known that, for such an $f$, any two of these three axioms imply the remaining one --- see \cite{aliprantis06}.} of the following properties:
\begin{enumerate*}
\item positive homogeneity;
\item convexity;
\item sub-additivity (the latter means that $f(X+Y)\leq f(X)+f(Y)$ for any $X,Y\in\X$).
\end{enumerate*}

\item\label{lo_ra_do} ({\sc Lower range dominance}) $f$ is \defin{lower-range dominated} if $\operatorname{domain}(f) \subseteq L^1$ and $f(X) \leq \E X - \mathrm{ess}\inf X =: \LR(X)$ for all $X$.

Lower range dominance is an essential property, as it reveals the interplay between coherent risk measures and generalized deviation measures --- see \cite{rockafellar06} for instance.

\item\label{law_inv} ({\sc Law invariance}) $f$ is \defin{law invariant} if $F_X = F_Y$ implies $f(Y) = f(X)$.

If $f$ is a risk functional, law invariance encapsulates the notion that, in appraising the risk of a position, we should only care about its statistical properties --- as these properties embody the uncertainty (w.r.t.\ the market outcome) faced by a given agent.
Law invariance is also important in empirical implementations, as it allows the theoretical risk measure to be estimated from historical data.

\item\label{lo_sem} ({\sc Lower-semicontinuity}) $f$ is \defin{lower-semicontinuous} if the set $ \Acc{k}{f}$ is closed, for all real $k$.

In the case when $\X$ is a metric space, lower-semicontinuity is equivalent to the following property: given any convergent sequence $ \{X_n\} \subseteq \X$, it holds that $f(\lim X_n) \leq \liminf f(X_n)$.

The \defin{convex envelop} of a mapping $f\colon\X\to\R$
%, where $C$ is a nonempty, closed and convex subset of $\X$,
is defined to be the extended real valued function $\conv f$ given by $\conv f(X)\coloneqq \sup_g g(X),\:X\in \X$,
where the supremum runs through all afine, continuous $g\colon \X\to\R$ satisfying $g\leq f$. Note that $\conv f$ is convex and lower-semicontinuous.

\item ({\sc Upper-semicontinuity}) $f$ is \defin{upper-semicontinuous} if the set $ \{X \in \X \colon\, f(X) \geq k\} $ is closed for all real $k$.

In the case when $\X$ is a metric space, upper-semicontinuity is equivalent to the following property: given any convergent sequence $ \{X_n\} \subseteq \X$, it holds that $f(\lim X_n) \geq \limsup f(X_n)$. Note that a functional $f$ is continuous if and only if it is both upper- and lower-semicontinuous.

%\item ({\sc Symmetry}) $f$ is \defin{symmetric} if $ f(X) = f(-X)$ for all $X\in\X$.
%
%In the theory of topological vector spaces, symmetry is one of the \emph{sine qua non} conditions in defining a seminorm; indeed, a seminorm is precisely the \pl\ of a symmetric, convex, and absorbing set. Although some deviation measures --- the standard deviation, for example --- do enjoy this attribute, symmetry is not really something desirable in our framework. Indeed, we are interested in quantifying the downside risk or deviation of a position, and thus dispersion above the mean can even be considered as ``good''. %{\color{red} AQUI TALVEZ VALHA A PENA CITAR ALGUMA COISA SOBRE ASYMMETRIC NORMS}

\item ({\sc comonotone additivity}) $f$ is \defin{comonotone additive} if $f(X +Y) = f(X) + f(Y)$ for every pair $X,Y\in\X$ such that $X$ and $Y$ are comonotone.

Comonotone additivity implies that a comonotone pair does not yield a gain, nor a loss, in diversification. This property sums up the notion that, for such a pair, an agent should be indifferent about how the two positions are kept, whether they are held in the same portfolio or separately. 

\end{enumerate}
\end{defi}

The \pl\ introduced in \cref{minkowski-functional-def} is the main tool used in this paper. Below, we recall its definition. We also introduce the cogauge, which is a straightly related dual concept. It is important to have in mind that, in the field of convex analysis, the \pl\ is known as the Minkowski gauge, Minkowski functional or, simply, ``the gauge''.

\begin{defi}\label{induce deviation} \label{SID}
Let $ A \subseteq \X $. The \defin{\pl\ of $A$} is the functional $\f \colon\X \rightarrow \ra$ defined, for $X\in\X$, by
\begin{align}\label{eq:fA}
\f (X) \coloneqq \inf \left\{m \in \R_+^*\colon\, m^{-1}{X} \in A \right\},
\end{align}
where $\inf \varnothing = \infty $.
The \defin{cogauge of $A$} is the functional $ \cog_{A}\colon\X \rightarrow \ra$ defined, for $X\in \X$, by
\begin{align}
\cog_{A} (X) \coloneqq \sup \left\{m \in \R_+^*\colon\, m^{-1}{X} \in A \right\},
\end{align}
where $\sup \varnothing = 0 $.\footnote{In the present setting, the convention $\sup \varnothing = 0$ is a sensible one, as we are taking the supremum over some subset of $(0,\infty)$.} 
\end{defi}

A financial interpretation is that the \pl\ answers the following question: given a set $A$ of acceptable positions, how much should we shrink (or ``gauge'') a certain position $X$ for it to become acceptable? The value $\f(X)$ is the required amount of shrinkage. This provides a limit to how leveraged can the position be. Notice that the following inclusions always hold:
\[
\{X \in \X\colon\, \f (X) < 1 \} \subseteq \A \subseteq \Acc{1}{\f}.
\]
The cogauge, in turn, is a useful concept that is closely linked the \pl: if we take a set $A$ comprised of non-acceptable positions, then the cogauge gives the most that we can shrink a position while keeping it non-acceptable. Importantly, for a star-shaped set $A$, gauge and cogauge are linked by the identity $\f = \cog_{A^\complement} $; see \Cref{coro cogauge}. For more details on cogauges, we refer the reader to \cite{rubinov86, rubinov00, zaffaroni08, zaffaroni13} and references therein.

\medskip

\section{Deviation Measures}\label{sec:deviations}
In this section we explore the functional $\f$ as measuring the amount of shrinkage on a financial position required to accommodate it in the base set $A$ of acceptable positions. Our focus here is the passage ``from the set to the measure''. Specifically, we present results that elucidate how attributes of the underlying set $A$ translate into mathematical and financial properties of the implied measure $\f$.

Before proceeding, let us introduce some further terminology. A non-negative and translation insensitive functional $D\colon \X \rightarrow \ra$ is called a \defin{deviation measure}; if, moreover, $D$ is convex, then it is said to be a \defin{convex deviation measure}.
{Non-negativity and translation insensitivity are taken as axioms in defining deviation measures because they capture, respectively, the intuitions that 
\begin{enumerate*}[label = (\roman*)]
\item a position whose payoff does not depend on the market outcome should display zero dispersion, and;
\item adding a fixed amount of cash to a given position should not alter its ``degree of non-constancy''.
\end{enumerate*}}
A positive homogeneous, convex deviation measure is said to be a \defin{generalized deviation measure}. Notice that the sub-level set $\Acc{k}{D}$ of a deviation measure $D$, for $k\ge0$, is never empty --- {indeed, it contains at least the set of all constant positions}. Of course, we say that $D$ is \emph{law invariant}, \emph{$\preceq$-monotone}, \emph{comonotone additive}, \emph{lower-range dominated}, etc., if it fulfills the corresponding properties as defined in the preceding section.

While in the pursuit of generality, we let any set $A$ be an acceptance set, a cornerstone property is star-shapedness, as we follow the rationale that shrinking an acceptable position yields a position which is still acceptable, and that $0$ (i.e., holding nothing) is also acceptable. The following lemma shows that demanding star shapedness is of no real consequence.

\begin{lemma}\label{star shaped hull} If $0 \in A$, then $\f = \D_{\sta(A)}$.
\end{lemma}

\begin{proof}
% Let $r_X \subset \R$ be implicitly defined for each $X$ by $(R_X \cup \{0\} ) \cap A$. Then:
% \begin{align*}
% \f (X) &= \inf \left\{m \in \R^*_+ : \dfrac{X}{m} \in A \right\}
% = \inf \left\{m \in \R^*_+ : \dfrac{X}{m} \in R_X \cap A \right\}
% \\ &= \inf \left\{m \in \R^*_+ : \dfrac{X}{m} \in r_X X \right\}
% = \inf \left\{m \in \R^*_+ : \dfrac{1}{m} \in r_X \right\}
% \\ &= \inf \left\{m \in \R^*_+ : 1 \in m r_X \right\} 
% {\color{red}= \inf \left\{m \in \R^*_+ : 1 \in \lambda m r_X, \text{for any}\lambda \in [0,1] \right\}}
% \\ &= \inf \left\{m \in \R^*_+ : 1 \in m [0,1]r_X \right\}
% = \inf \left\{m \in \R^*_+ : 1 \in m \; \sta(r_X) \right\}
% \\ &= \inf \left\{m \in \R^*_+ : \dfrac{X}{m} \in R_X \cap \sta( A) \right\}
% = \inf \left\{m \in \R^*_+ : \dfrac{X}{m} \in \sta ( A ) \right\}
% \\ &= \D_{\sta(A)}(X)
% \end{align*}
% as stated.
Clearly $\f(0)=\D_{\sta(A)}(0)$. Fix, then, $0\ne X\in\X$, and let $T\colon \R_+^* \to \X$ be defined through $T(m) = X/m$ for $m>0$. Write \[\tilde{m} = \D_{\sta(A)}(X) = \inf T^{-1}\big(\sta(A)\big)\quad \text{and}\quad \hat{m} = \f(X) = \inf T^{-1}(A).\] Clearly $\tilde{m}\le\hat{m}$ as $T^{-1}\big(\sta(A)\big)\supseteq T^{-1}(A)$. It remains to show that $\hat{m}\le\tilde{m}$, or, which is to say the same, that $\hat{m}$ is a lower bound for the set $T^{-1}\big(\sta(A)\big)$. Let, therefore, $m\in T^{-1}\big(\sta(A)\big)$, which means that $X/m\in \sta(A)$ and which, by definition, occurs if and only if $X/m = yZ$ for some $y\in[0,1]$ and some $Z\in A$. Then, as $y\ne0$ since $X\ne0$, we have $X/(ym)\in A$ and it follows that $\hat{m}\le ym \le m$.
\end{proof}

We are now interested in controlling the variability of a financial position. For such, we have a given acceptance set $A$ which contains only positions whose ``deviations'' are deemed acceptable. Hence, there are some natural proprieties that $A$ should posses; arguably, the most fundamental property is that it should be insensitive to addition of a constant, i.e., $A$ should be stable under scalar addition. Another basal property that should be required is that any position that has positive risk (remember that here ``risk'' is exclusively associated with variability) should not be allowed to be arbitrarily expanded, as the dispersion ought to increase together with size. Positions that have no risk, in turn, are allowed to be expanded arbitrarily and, considering that we reckon only constants as riskless, we see that the attribute we desire is that $A$ be radially bounded at non-constants. The next proposition shows that whenever we appraise an acceptance set which is star-shaped, radially bounded at non-constants, and stable under scalar addition, its implied \pl\ is, not surprisingly, a deviation measure.

\begin{propo} \label{radially}
Let $\A \subseteq \X$ be star-shaped. Then the following holds
\begin{enumerate}[label = (\roman*)]
\item\label{radially item 1} If $A$ is radially bounded, then $\f (X) > 0$ for all $X \in \X \setminus \{0\} $ and $\D_{A \cup \R} $ is non-negative. Therefore, $\f$ is non-negative whenever $A$ is radially bounded at non-constants.
\item\label{radially item 2} If $A$ is stable under scalar addition, then
$\f$ is translation insensitive.
\end{enumerate}
In particular, if $A$ is star-shaped, radially bounded at non-constants and stable under scalar addition, then $\f$ is a deviation measure. 
\end{propo}

\begin{proof}
For the first item, notice that if $\A$ is radially bounded, then $\f (X) > 0$, for every non-zero $X \in \X$. Indeed, if $\A$ is radially bounded then --- by definition --- for each $X$ there is a $m_X > 0$ such that $m^{-1} X \notin \A$, for all $m < m_X$. Therefore, it holds that $\inf \{m \in \R_+^* \colon\, m^{-1}X \in \A \} >0$. Now, observe that, as $\R$ is a subspace of $\X$, one has $\D_\R (X) = 0$ for every $X \in \R$, whereas $\D_\R (X) = \infty$, for each (a.s.) non-constant $X$ (see \Cref{star}). Then, as it is easily seen that $\D_{B\cup B'} = \min\big(\D_B,\D_B'\big)$ for any non-empty sets $B$ and $B'$, we have $\D_{A \cup \R} (X) = \min (\f (X) , \D_\R(X)) = \f(X) > 0$, for every $X \notin \R$, whereas for $c \in \R$ we have that $\D_{A \cup \R} (c) = \min (\f (c) , \D_\R(c)) = \min (\f (c) , 0) = 0$.

For \Cref{radially item 2}, notice that star-shapedness together with stability under scalar addition yield $ \R \subseteq A$. Thus, we clearly have $\f(c) = 0$ for every $c \in \R$. It is also clear that for such a $c$ one has $X + c \in A$ if and only if $X \in A$. In particular the condition $(X+c)/m \in A$ is equivalent to $X/m \in A$, hence
\[
\f (X + c) = \inf \{m > 0 \colon\, (X+c)/m \in A\} = \inf \{m > 0 \colon\, m^{-1}{X}\in A\} = \f(X),
\]
for any $c\in\R$.
\end{proof}

\begin{remark}
It is possible to extend the notion of ``risklessness'' from only $\R$ to an arbitrary cone $B$, in which case radial boundedness at non-constants should be replaced by radial boundedness at $B^\complement$, whence $\f (X)$ would be greater than $0$ for any $X \notin B$, and $0$ for $X \in B$. Furthermore, stability under scalar addition should be replaced by stability under addition of members of $B$, yielding then $\f (X +b) = \f(X)$ for $b \in B$.
\end{remark}

\begin{remark}
\cite{farkas19} proposed a novel way to measure risk, the intrinsic risk measure, defined as $\mathfrak{D}_{A} (X) = \inf \{ \lambda \in [0,1] : (1-\lambda) X+ \lambda\frac{ \pi (X)}{\pi(S)} S \in A\}$, where $S$ is an eligible asset and $\pi(X)$ represents the price of $X$. If $A$ is sable under addition of multiples of $S$ we have, for all $X \notin A$ the relationship\footnote{Under the convention $\frac{1}{\infty} = 0$.} $1-\mathfrak{D}_A(X)= (\f(X))^{-1}$. This gives us that the \pl \; provides all information that $\mathfrak{D}_A$ can provide. However, the reciprocal is not true, as for any $ X \in A$  the intrinsic risk measure is stuck on $0$.
\end{remark}

We now present some results regarding the set $A+\R$, seen as the result of an operation $A\mapsto A+\R$ taken on some basis set $A$. It is particularly interesting because first, it coerces an arbitrary set to become stable under scalar addition and, secondly, it seamlessly harmonizes with the notion of \emph{measures of error} (see \Cref{quadrangle}). 

\begin{lemma}\label{A+R deviation}
Let $A\subseteq \X$ be non-empty. Then $A + \R$ is stable under scalar addition. Assuming further that $A$ is star-shaped, closed and radially bounded we have that $A+\R$ is radially bounded at non-constants.
\end{lemma}

\begin{proof}
The first claim is obvious. The second claim holds by \Cref{star-shaped no cone is radially bounded}, as in this case $A$ contains no proper cone with vertex at some $x \in \R$. Hence, $A + \R$ contains no cones other than $\R$ and $\{0\} $, that is, $A+\R$ is radially bounded at non-constants. 
\end{proof}

\begin{remark}
Given an arbitrary star-shaped, closed and radially bounded set we have that $\D_{A + \R}$ is a deviation measure. An (apparent) sensible choice for the acceptance set $A$ would be a sub-level set $\Acc{k}{\rho}$ corresponding to some pre-specified coherent risk measure $\rho$ and $k \in \R$. However, such a set is never radially bounded. Nevertheless, if we insist on taking $B \coloneqq \Acc{k}{\rho} + \R$ in order to force translation insensibility, then we would have that $B \equiv \{X \in \X \colon\, \rho < \infty\} $, which again is of no interest as it is clearly a cone, with $\D_B (X) = 0$ for $X \in B $ and $\D_B (X) = \infty$ otherwise. Said another way, in this case $\D_B$ is the characteristic function of $\rho$. 
\end{remark}

\begin{propo}\label{set oper-item4}
Let $A, B \subseteq \X$ and assume $B$ is a cone. Then, for each $X \in \X$, it holds that $\D_{A+ B} (X) = \inf_{Z \in B} \f (X-Z)$.
\end{propo}

\begin{proof}
If $B = \{0\}$ there is nothing to show. If $B$ is a proper cone, let $m>0$. Then one has $m^{-1}X \in \A+ \B$ if and only if $m^{-1}{X} = a + b$ for some $a \in \A$ and some $b \in \B$, if and only if $m^{-1}X-b = a$, for some $b \in B$ and some $a\in \X$ such that $\f(a) \leq 1$, if and only if $\f \left(m^{-1}X -b \right) \leq 1$ for some $b\in B$. By positive homogeneity, the latter sentence is equivalent to the following: there exists a $b\in B$ such that $\f \left(X -mb \right)\leq m$. Additionally --- as $B$ is a cone --- if there is an element $b\in B$ that respects $\f (X -mb ) \leq m$. Then by letting $d = mb$ we see that there is an element $d\in B$ such that $\f (X - d ) \leq m$, and the reciprocal of the previous sentence is obviously also true: that is, it holds that $\f (X - mb)$ for some $b \in B$ if and only if $\f (X - d ) \leq m$ for some $d \in B $. In view of the above equivalences, by writing $M_b\coloneqq \{m\in\R_+^*\colon\, \f (X-b)\leq m\} $ and noticing that $\f (X-b) = \inf M_b$, we finally have that
\begin{align*}
\D_{A+B}(X) &= \inf \bigcup\nolimits_{b\in B} M_b\\
&= \inf\nolimits_{b\in B}\inf M_b
\end{align*}
as asserted.
\end{proof}

\begin{remark}\label{quadrangle}
\cite{rockafellar13} proposed \emph{measures of error} to quantify the ``non-zeroness'' of a random variable. {By definition, a functional $\varepsilon\colon \Lp \rightarrow \ra $ is called a \defin{measure of error} if it is lower-semicontinuous, sub-linear, positive homogeneous and satisfies
\begin{enumerate*}[label = (\roman*)]
\item $\varepsilon (X) = 0 $ if and only if $X=0$ almost surely; and
\item if $\lim \varepsilon (X_n) =0$ then $ \lim \E X_n = 0 $
\end{enumerate*}
.} By the authors' Quadrangle Theorem, {if $\varepsilon$ is a measure of error, then the functional $D$ defined, for $X\in\X$, by} $D (X) \coloneqq \min_{c \in \R}\, \varepsilon (X-c)$, is a convex deviation measure. Furthermore, such $D$ is a generalized deviation measure whenever the inequality $\varepsilon (X) \leq | \E X|$ holds for every $ X \leq 0 $. From this we can conclude that, given a functional $\varepsilon$ satisfying all those conditions, the identity $D(X) = \D_{(\Acc{1}{\varepsilon} + \R)} (X)$ holds. Indeed, if the minimum is attained, it holds that $\D_{(\Acc{1}{\varepsilon} + \R)} (X) = \inf_{c \in \R} \D_{\Acc{1}{\varepsilon}} (X - c) = \min_{c \in \R} \varepsilon (X - c) = D(X)$.

\end{remark}

\begin{remark}
In the context of \Cref{set oper-item4}, we have from \Cref{star}, \cref{lemma-item3}, that $\D_B(X)=0$ for any $X\in B$ since $B$ is a cone. Thus, $\D_{A+B} (X)=\inf_{Z\in B}
\{\f(X-Z)+\D_B(Z)\} = \inf_{Z\in \X} \{\f(X-Z)+\D_B(Z)\} $. The last equality holds because, for any $Z \notin B$, as $B$ is a cone, it follows by \Cref{star} that $\D_B(Z) = \infty$. This concept is closely related to inf-convolution and optimal risk sharing. Inf-convolution is a well-known operation for functionals in convex analysis --- for details of the use of inf-convolution in risk share we refer the reader to \cite{barrieu05}, \cite{jouini08} and \cite{righi20}. 
\end{remark}

We now proceed with our investigation of some of the more prominent set theoretical properties found in the literature. First, let us consider the \emph{principle of diversification}, which assets that any convex combination of acceptable positions should be acceptable as well. Obviously, this corresponds to the formal requirement that the acceptance set be convex. The next result provides a sufficient condition which ensures that the \pl\ is a generalized deviation measure. Note that, by construction \pl\ is always positive homogeneous.

\begin{propo}\label{generalized dev} 
Let $0 \in A\subseteq\X$. The following assertions hold
\begin{enumerate}[label = (\roman*)]

\item \label{propo2-item p.h} $\f$ is positive homogeneous.

\item\label{coro conv item1} If $A$ is convex, then $\f$ is sub-linear.

\end{enumerate}
In particular if $A$ is convex, radially bounded at non-constants, stable under scalar addition and contains the origin, then $\f$ is a generalized deviation measure.
\end{propo}

\begin{proof}
For the first item, clearly $\f(0X) = \f(0) = \inf\{m\in\R_+^*\colon\,m^{-1}0\in A\} = \inf \R_+^* = 0 = 0 \f(X)$. Moreover, given $ \lambda>0$, we have
\begin{align*}
\D_{\lambda \A} (X) &= \inf \left\{m \in \R_+^*\colon\, m^{-1}{X} \in \lambda \A\right\}
= \inf \left\{m \in \R_+^*\colon\, (\lambda m)^{-1}X \in \A\right\}\\
&\quad = \inf \left\{{m}{\lambda^{-1}} \in \R_+^*: m^{-1}{X} \in \A\right\}
= \lambda^{-1}{\f (X)}
\end{align*}

For the second item, it is enough to show that $\f$ is convex, so fix $\lambda\in[0,1]$ and $X,Y\in \X$. Define
\[
\mathfrak{A} \coloneqq \lbrace\alpha\in\R_+^*\colon\, \lambda X\in \alpha A\rbrace \quad\mbox{ and }\quad \mathfrak{B} \coloneqq \lbrace\beta\in\R_+^*\colon\, (1-\lambda) Y\in \beta A\rbrace.
\]
By definition and positive homogeneity we have $\inf\mathfrak{A} = \f (\lambda X) = \lambda \f(X)$ and $\inf\mathfrak{B} = \f((1-\lambda)Y) = (1-\lambda) \f(Y)$.
We only need to consider the case where both $\mathfrak{A}$ and $\mathfrak{B}$ are non-empty, as otherwise the upper bound $\f(\lambda X+(1-\lambda)Y)\leq \infty$ holds trivially.
Take $\alpha\in\mathfrak{A}$ and $\beta\in\mathfrak{B}$. Then, convexity of $A$ yields $\lambda X+(1-\lambda)Y \in (\alpha + \beta)A$, and hence $\f (\lambda X + (1-\lambda)Y) \leq \alpha + \beta$.
Therefore, $\f (\lambda X + (1-\lambda)Y)\leq \inf\mathfrak{A} + \inf\mathfrak{B} = \lambda\f(X) + (1-\lambda)\f (Y) $.

\end{proof}

Our approach stresses the importance of appraising the underlying pool of positions deemed acceptable as a fundamental building block in the quantification of risk. That is, we emphasize the passage ``from the set to the measure''. At any rate, an essential part of any such approach is the study of the interplay between acceptance sets and their corresponding measures. In particular, it is desirable to be able to ``recover one from another''. The following proposition goes in that direction.

\begin{propo}\label{lemma-item4}
It holds that $\{\f < 1\} \subseteq \sta (A) \subseteq \Acc{1}{\f}$ and, if $A$ is closed and star shaped, $A = \Acc{1}{\f}$. Furthermore, a functional $D\colon \X \to \ra$ is positive homogeneous if and only if $D = \D_{\Acc{1}{D}}$ with $0 \in {\Acc{1}{D}}$.
\end{propo}

\begin{proof}
\Cref{star shaped hull}, together with Lemma 5.49 of \cite{aliprantis06}, yields the first assertion. For the second, the if direction follows straightforwardly from \Cref{generalized dev}, \cref{propo2-item p.h}. The only if direction, in turn, follows from \[ \D_{\Acc{1}{D}} (X) = \inf \left\{m \in \R_+ : \dfrac{X}{m} \in \Acc{1}{D} \right\} = \inf \{m \in \R_+ : D(X) \leq m \} = D(X), \] and, noting that $D(0)=0$, we see that $0 \in \Acc{1}{D}$.
\end{proof}

The next proposition tell us that our \pl is a reasonable good approximation of a non-positive homogeneous convex deviation measure. In particular they agree on acceptability.

\begin{propo}
 Assume that $f$ is a convex deviation measure that is not positive homogeneous and that $f(0)=0$ and let $A := A^k_f$. Then
 
\begin{enumerate}[label = (\roman*)] 

\item $\f $ is a generalized deviation measure.

 \item If $f$ is lower semi continuous $A = A_\f^1$. Furthermore,  $ k \f(X) \geq f(X)$ for all $ X \in A$  and $ k \f(X) \leq f(X)$ for all $ X \notin A$.
 
  \end{enumerate} 
\end{propo}

\begin{proof}

That $\f$ is a generalized deviation measure is obvious from our previous results as stability under scalar addition is not affected by positive homogeneity and as convexity yields $f(\lambda X) \geq \lambda f(X)$ for big enough $\lambda$ radially bounded at non-constants is not affected.

 For $A =A^1_\f$, without loss of generality, let $k=1$, first note that  $X \in A$ it implies that $\f(X) \leq 1$, which in turn implies that $X \in A^1_\f$, hence, $A \subseteq A^1_\f$. Now, let $ X \in A^1_\f$ which gives us  $ 1 \geq \f(X) = \inf \{ m > 0 : \frac{X}{m} \in A\} = \inf \{ m > 0 : f \left( \frac{X}{m} \right) \leq 1\} $, and this implies in $1 \geq f \left( \frac{X}{\f(X)}\right) \geq \frac{f(X)}{\f(X)} \geq f(X)$, therefore, $A^1_\f \subset A$.

 Lastly, the inequalities holds trivially for constants, therefore we shall prove only for non-constants (those which $\f$ and $f$ are strictly greater than $0$). As $A$ is closed, the infimum is always attained (Proposition 3.1 item vi), hence $f \left(\frac{X}{\f(X)}\right) =1$ then  by convexity of $f$, $1 = f \left(\frac{X}{\f(X)}\right) \geq \frac{f(X)}{\f(X)}$ if $X \in A$ and the opposite holds if $X \notin A$,  $1 = f\left(\frac{X}{\f(X)}\right) \leq \frac{f(X)}{\f(X)}$.
 
\end{proof}

%In the next lemma we explore some basic relations between the underlying sets and the corresponding \pls.
\subsection{Optimization and continuity}

The next result deals with the solution of the minimization problem appearing in the definition of the \pl.

\begin{propo} \label{propo2}
Let $\A \subseteq \X$ be non-empty. Then, we have the following:
\begin{enumerate}[label = (\roman*)]
\item \label{propo2-item1} If $\A$ is absorbing, then $\f$ is finite-valued.
\item \label{propo2-item2}If $A$ is star-shaped and $y\coloneqq \f (X) \in \R_+^*$, then $y^{-1}X\in\bd(\A)$, i.e.,\ $ \f (X) = 1 $ implies $X$ lies in the boundary of $\A$.
\item \label{propo2-item3} If $A$ is strongly star-shaped, $X \in $ $\bd(A)$ and $X \neq 0$ then, $\f (X) = 1$.
\item \label{propo2-item5} If $\f (X) \in \R_+^*$, then
\[
\f (X) = \inf \{m \in \R_+^* \colon\, m^{-1}{X} \in \A \} = (\sup \{m \in \R_+ \colon\, mX \in \A \} )^{-1}.
\]

\item \label{propo2-item6}If $R_X \cap A = \varnothing $ then $\f (X) = \infty$. In particular if $0 \notin A$ then $\f(0) = \infty$.
\item \label{propo2-item 7} If $ A $ is closed, absorbing and radially bounded, then the {infimum in \Cref{eq:fA} is attained for any $X \in \X \setminus \{0\} $, that is,
\(
X\in \f (X)A
% \inf \left\{m \in \R_+^*\colon\, m^{-1}X \in A \right\} = \min \left\{m \in \R_+^*\colon\, m^{-1}X \in A \right\} > 0,
\)
for any $X \in \X $.}

\item If $A$ is closed, then the infimum in \Cref{eq:fA} is attained for any $X$ such that $\f (X) \in \R_+^*$.
\end{enumerate}
\end{propo}

\begin{proof}
Let $X\in\X$ and write $y \coloneqq \f(X)$.

For the first item, there exists --- by the absorbing property --- some $\delta_X \in \R_+^*$ such that the inclusion $[0,\delta_X] X \subseteq A $ holds. It is straightforward to see that in this case the set $ \left\{m \in \R_+^* \colon\, m^{-1}{X} \in A \right\}$ is never empty. Therefore, $\f(X)<\infty$.

For the second item, it suffices to consider the case $\f(X)=1$, as the general case then easily follows from positive homogeneity. In order to verify that $X\in\bd(A)$, we only have to exhibit sequences $\{Y_n\}\subseteq A$ and $\{Z_n\}\subseteq A^\complement$ such that $\lim Y_n = \lim Z_n = X$. Let, then, $Y_n$ be defined through $Y_n \coloneqq (1 + 1/2^n)^{-1}X$ and, similarly, $Z_n \coloneqq (1-1/2^n)^{-1}X$. Continuity of scalar multiplication immediately yields the desired equality of limits, so it only remains to show that $Y_n\in A$ and $Z_n\in A^\complement$ for all $n$. For such, just notice that --- due to star-shapedness through \Cref{star} --- if $m > 1$, then $m^{-1}X \in A$, so $Y_n\in A$, and if $0<m<1$, then $m^{-1}X\notin \A$, so $Z_n\notin A$.

For \Cref{propo2-item3}, as $X\neq0$ by assumption, we see that whenever the ray $R_X \equiv \{\lambda X \colon\, \lambda > 0\} $ has a non-empty intersection with $\bd(A)$, it necessarily also holds that $\f (X) \in \R_+^*$. Therefore, as $X \in\bd(A)$, we also have that $\f (X)^{-1}X \in\bd(A)$, by \Cref{propo2-item2}. Thus, we have $X\in R_X\cap \bd\,A$ and $\f (X)^{-1}X\in R_X\cap\bd\, A$, and hence strong star-shapedness of $A$ tells us that $\f(X) = 1$.

For \Cref{propo2-item5}, notice that
\[
y = \inf \left\{m \in \R_+^* \colon\, m^{-1}{X} \in \A \right\} = \inf \left\{{m}^{-1} \in {\R_+^*} \colon\, mX \in \A \right\}.
\]
Now, if $x^{-1}>0$ is a lower bound for the set $\{m^{-1}\in \R_+^*\colon\,mX\in A\} $, then $x$ is an upper bound for the set $\{m\in \R_+\colon\,mX\in A\} $; if $x^{-1}$ is the largest such lower bound, then $x$ is the smallest such upper bound. That is to say, one has $ y^{-1} = \sup \{m \in \R_+ \colon\, mX \in \A \} $.

\cref{propo2-item6} is clear as if $\{\lambda X \colon\, \lambda >0 \} \cap A = \varnothing $ then the set $\{m \in \R^*_+\colon\, m^{-1}X \in A\} $ is empty and the infimum of such set are $\infty$.

For item \Cref{propo2-item 7}, notice that if $\A$ is radially bounded, then by \Cref{radially} \Cref{radially item 1}, $\f (X) > 0$, for every non-zero $X \in \X$. Now, let $T_X\colon\R_+^*\to\X$ be defined by $T_X(m) = m^{-1}X$. Clearly, $T_X$ is continuous. Thus, if $A$ is a closed subset of $\X$ so is $T_X^{-1}(A)$ a closed subset of $\R_+^*$. Also, if $A$ is absorbing, then $T_X^{-1}(A)$ is non-empty. Finally, since radial boundedness ensures $\f(X)>0$, it follows that $\inf T_X^{-1}(A) \in T_X^{-1}(A)$ as stated.

The proof of item the last item is identical to the previous one.
\end{proof}

We now address the continuity of \pl
\begin{propo}\label{strongly continuidade}
If $A$ is star-shaped and closed, then $\f$ is lower-semicontinuous. If, additionally, $A$ is strongly star-shaped, then we have the following: If $0 \in \bd(A)$, then $\f$ is continuous except at $0$. If $ 0 \in \interior (A)$, then $\f$ is continuous everywhere.
\end{propo}

\begin{proof}
To show that $\f$ is lower-semicontinuous it is enough to show that $\Acc{1}{\f}$ is closed, but then by \Cref{lemma-item4} we have $\Acc{1}{\f} = A$ which is closed by assumption. Now, let $A$ be strongly star-shaped, so in particular we have $0 \in A$. We will show first the case $0 \in \bd(A)$, and then consider the case $0 \in \interior (A)$.

Assume then that $0 \in \bd (A)$. Note that if we let $B \coloneqq A \setminus \{0\} $, then it is an easy check to see that $\D_B(X) = \f(X) $ for all $X \in \X \setminus \{0\} $ and $\D_B(0) = \infty$. Indeed, for $X\neq0$ the conditions $m^{-1}X\in A$ and $m^{-1}X\in B$ are clearly equivalent, whereas for $X=0$ the condition $X\in mA$ is always true whereas $X\in mB$ is vacuous.
Hence, we have that $\D_B$ is lower-semicontinuous everywhere, except at $0$. Furthermore, we have $B = \Acc{1}{\D_B}$. To see it, note that $\Acc{1}{\D_B} = \{X \in \X\colon\, \D_B (X) \leq 1\} $, and obviously $0 \notin \Acc{1}{\D_B}$ as $\D_B(0)=\infty$. Therefore, \begin{align*}
\Acc{1}{\D_B} &= \Acc{1}{\D_B} \setminus \{0\} 
\\ &= \{X \in \X \colon\, \D_B (X) \leq 1\} \setminus \{0\} 
\\ &= \{X \in \X \setminus \{0\} \colon\, \D_B (X) \leq 1\} 
\\ &= \{X \in \X \setminus \{0\} \colon\, \f (X) \leq 1\} 
\\ &= \{X \in \X \colon\, \f (X) \leq 1\} \setminus \{0\} 
\\ &= B.
\end{align*}
It remains to show that the set $ V \coloneqq \{X \in \X \colon\, \D_B <1\} $ is open, from which we will know that $\D_B$ is upper-semicontinuous. This will give us then that $\D_B$ is continuous everywhere except at $0$, which in turn entails continuity of $\f$ everywhere except at $0$. To see that $V$ is indeed an open set, note that --- due to \Cref{propo2} \cref{propo2-item2,propo2-item3,propo2-item6} --- if $X \neq 0 $ then $X \in \bd (A) $ if and only if $\f(X) = 1$, hence $\bd (A) = \{X \in \X\colon\, \f(X) \equiv \D_B(X) = 1\} \cup \{0\} $. 
Now, the reader should realize that, again since $\D_B(0) = \infty$,
\begin{align*}
V &= \{X \in \X \colon\, \D_B(X) <1 \text{ and } X\neq0\} \\
&= \{X \in \X \colon\, \f(X) <1 \text{ and } X\neq0\} \\
% &= \{X \in \X \colon\, \D_B(X) <1\} 
% \\ &= \{X \in \X \colon\, \D_B(X) <1\} \setminus \{0\} 
% \\ 
%\\ &= \{X\setminus \{0\} \in \X \colon\, \f(X) <1\} 
%\\ &= \{X \in \X \colon\, \f(X) <1\} \setminus \{0\} 
&= (A \setminus \bd (A) ) \setminus \{0\} \\
&=A \setminus \bd (A) \\
&= \interior (A),
\end{align*} and as $\interior (A)$ is by definition an open set, the claim that $V$ is open holds.

Finally, if $0 \in \interior (A)$, as we already have that $\f$ is lower-semicontinuous, it is enough to show that it is also upper-semicontinuous. It suffices to show that the set $U \coloneqq \{X \in \X \colon\, \f <1\} $ is open.
Clearly, again due to \Cref{propo2} \cref{propo2-item2,propo2-item3,propo2-item6}, we have $\bd(A) = \{X \in \X\colon\, \f(X) = 1\} $. Hence, $\interior (A) = A \setminus \bd(A) = U$ and the claim follows.
\end{proof}

\subsection{Comonotone additivity and concavity}
 We now develop a characterization of comonotone additivity from the perspective of acceptance sets. A financial intuition of comonotone additivity is the following: that two comonotone positions do not provide neither diversification benefit nor brings harm to the portfolio. This attribute manifests in the risk/deviation measure as an indifference between the sum of the risk of two comonotone positions, on the one hand, and the risk of their sum on the other. In the acceptance set $A$, the intuition that diversification (with comonotone pairs) brings no benefit is translated as $A^\complement$ being convex for comonotone pairs, and the idea that diversification (again, with comonotone pairs) brings no harms turns in $A$ being convex for comonotone pairs. The first step now is to find conditions ensuring $\f$ is concave, and then the conditions for it to be additive.

\begin{propo}\label{coro conv}
If $A$ is star-shaped and $A^\complement$ is convex, then $\f$ is super-linear (concave and positive homogeneous) on $\cone (A^\complement)$, i.e., $\f(X + Y ) \geq \f(X) + \f(Y)$ for any $X,Y \in \cone (A^\complement)$.
\end{propo}

\begin{proof}
We already have positive homogeneity by \Cref{generalized dev}, as $0\in A$. Star-shapedness of $A$ and \Cref{coro cogauge} tell us that $\f = \cog_{A^\complement}$. Hence, it suffices to show that $\cog_{A^\complement}$ is a concave functional on $\cone (A^\complement)$ whenever $A^\complement$ is convex. To see that this is the case, let $B = A^\complement$, and fix $\lambda\in[0,1]$ and $X,Y\in \cone (A^\complement)$.
Let us first consider the case where $0<\lambda<1$ and where both $X$ and $Y$ are non-zero. In this scenario the sets
\(
\mathfrak{A} \coloneqq \{\alpha\in\R_+^*\colon\, \lambda X\in \alpha B\}
\)
and
\(
\mathfrak{B} \coloneqq \{\beta\in\R_+^*\colon\, (1-\lambda)Y\in \beta B\}
\)
are both non-empty (for instance, $X\in\cone(B)$ means precisely that $X = aZ$ for some $a>0$ and some non-zero $Z\in B$, and in this case we have $\lambda a \in \mathfrak{A}$). By definition and using positive homogeneity of $\f$ together with the equality $\f = \cog_B$, we have $\sup\mathfrak{A} = \cog_B(\lambda X) = \lambda \cog_B(X)$ and $\sup\mathfrak{B} = \cog_B((1-\lambda)Y) = (1-\lambda)\cog_B(Y)$. Taking $\alpha\in\mathfrak{A}$ and $\beta\in\mathfrak{B}$, convexity of $B$ yields $\lambda X+(1-\lambda)Y \in (\alpha + \beta)B$, so $\cog_B (\lambda X + (1-\lambda)Y) \geq \alpha + \beta$. Therefore, $\cog_B (\lambda X + (1-\lambda)Y) \geq \sup\mathfrak{A} + \sup\mathfrak{B} = \lambda \cog_B (X) + (1-\lambda)\cog_B (Y) $. The remaining cases are just a matter of adapting the following argument: if, say, $\lambda X = 0$, then $\mathfrak{A}=\varnothing$ and $\cog_B(\lambda X + (1-\lambda Y)) = \cog_B((1-\lambda)Y) = (1-\lambda)\cog_B(Y) = \lambda\cog_B(X) + (1-\lambda)\cog_B(Y)$. This completes the proof.
\end{proof}

\begin{remark}
Unfortunately, \Cref{coro conv} cannot be relaxed as to accommodate
super-linearity of $\f$ on the whole $\X$. However, if we are willing to let go from the identity $\f = \cog_{A^\complement}$, it is possible to define the cogauge in a slightly different manner, by assigning the value $\cog_B(X)\coloneqq -\infty$ whenever $\{m \in \R_+ \colon\, m^{-1} X \in B\} = \varnothing $; in this case, an easy adaptation of the proof of \cref{coro conv item1} in \Cref{generalized dev} yields concavity of $\cog_B$ for convex $B$. This alternative definition of the cogauge was studied in \cite{barbara94}.
To see that the assumptions in \Cref{coro conv} do not, in general, yield super-linearity of $\f$ on the whole $\X$, consider the following counterexample, illustrated in \Cref{fig counter concave}: let $\Omega = \{0,1\}$ be the binary market and identify $L^0\equiv\R^2$ as usual. Let $A\coloneqq\{(x,y)\in\R^2\colon\, y-\vert x\vert\leq1\}$. In this case, the set $C\coloneqq A \setminus \cone (A^\complement)$ is a cone and hence, for any $X \in C$, we have that $\f(X) = 0$, whereas $\f(X)>0$ for $X\notin C$. Now let $Y = (1,\nicefrac12)\in \interior B$, $Z = (1,1)\in \bd B$ and $W = (1,2)\in\bd A$. We have 
$\f(Z) =0< \f(W)$, but $Z$ is a convex combination of $W$ and $Y$, so $\f$ is not concave on the whole domain.

\end{remark}

{Now consider a cone $C$ comprised of positions that do not provide any benefit or detriment from diversification. By a \emph{benefit} from diversifying a position $X$ with an asset $Y$ we mean that the risk, or dispersion, of the overall portfolio will not increase if we take a convex combination of $X$ and $Y$ when compared to any one of the individual positions. In the acceptance set, such reasoning is reflected by noting that if both $X$ and $Y$ are acceptable, then their convex combinations cannot be worse --- that is to say, convex combinations of acceptable positions are acceptable as well. This rationale says that the acceptance set $A$, or at least its positions also lying in $C$, should be a convex set, i.e., {we should require that $C \cap A$ be convex}. On the other hand, by a \emph{detriment} from diversifying $X$ with a position $Y \in C$, we mean the exact opposite: that the risk or dispersion of any convex combination of $X$ and $Y$ should not be less than the individual positions. With respect to an acceptance set $A$, this means that if both $X$ and $Y$ are not deemed acceptable ($X,Y\notin A$), then combining them in a convex fashion yields an unacceptable position as well. Hence, the complement of $A$ should be convex, at least when restricted to $C$: we should also require that $A^\complement \cap C$ be a convex set. Importantly, when restricted to such a cone, the \pl\ of a star-shaped set $A$ is linear:}

\begin{propo} \label{linear}
Let $A$ be a star-shaped set, and let $C \subseteq \cone (A^\complement)$ be a cone for which both $A \cap C $ and $A^\complement \cap C $ are convex sets. Then $\f $ respects $\f(X + Y) = \f(X) + \f(Y) $ for every $ X,Y \in C$.
\end{propo}

\begin{proof}
Let $g$ be the restriction of $\f$ to the cone $C$, i.e.,\ $g \colon C \rightarrow \ra$ is such that $g (X) = \f(X) = \max\big(\f(X), \D_C(X)\big) = \D_{A \cap C} (X)$ for all $X \in C$. It suffices to show that $g$ is additive; we shall proceed by showing that this function is concave and sub-linear. Sub-linearity of $g$ is yielded by \cref{coro conv item1} of \Cref{generalized dev}, as $A \cap C $ is a convex set containing the origin by assumption, and thus $\D_{A\cap C}$ is sub-linear on the whole $\X$, in particular when restricted to $C$. For concavity, we shall summon the cogauge to help us: as $A $ is a star-shaped set, the gauge coincides with the cogauge of its complement, i.e.,\ $ \f = \cog_{A^\complement}$ --- see \Cref{coro cogauge}. It follows that, for $X\in C$, one has $g (X) = \cog_{A^\complement} (X)$.

We now show that, for $X\in C$, the identity $\cog_{A^\complement} (X) = \cog_{A^\complement \cap C} (X)$ holds. As $A\cup C^\complement$ is star-shaped since $C^\complement\cup\{0\} $ is a cone, we have
\[
\cog_{A^\complement \cap C} (X) = \cog_{(A\cup C^\complement)^\complement} (X)
= \D_{A \cup C^\complement} (X)
= \min(\f(X) , \D_{C^\complement} (X) )
= \min \big(\cog_{A^\complement}(X),\cog_{C}(X)\big).
\]
In particular, $g = \cog_{A^\complement\cap C}$ on $C$, as $\cog_C (X) = \infty = \D_{C^\complement} (X)$ if $X \in C$ and $\cog_C (X) = 0=\D_{C^\complement} (X)$ if $ X \notin C$.

Now, the only thing that is left is to show is that the cogauge of a convex set is a concave function on $C$, and this follows from \Cref{coro conv} as it tells us that $\cog_{A^\complement \cap C}$ is concave on $\cone (A^\complement) \supseteq C$.
\end{proof}

The preceding reasoning and results yield comonotonic additivity of $\f$ whenever $A$ and $A^\complement$ are both convex for comonotone pairs; this is the content of \Cref{coro como}. As an example of a set $A$ satisfying the assumptions in the corollary, take $\Omega = \{0,1\}$, identify $L^0\equiv\R^2$, and let $A$ be the set of those $X=(u,v)\in\R^2$ for which $u\geq0$, $v\geq0$ and $\vert u\vert + \vert v\vert \leq 1$. In this case, the set of comonotone pairs in the 1st quadrant is precisely $\{(u,v)\in\R_+^2\colon\,u\geq v\}$.

\begin{lemma}\label{lemma convex cone comono}
Let $X\in\X$. Then the family
\(
C_X \coloneqq \{Y \in \X\colon\, \text{$Y$ is comonotone to $X$} \} 
\)
is a convex cone {which is closed with respect to the topology of convergence in probability}. Furthermore, if $(X,Y)$ is a comonotone pair, then any two elements of the convex cone \(C_{X,Y} \coloneqq \conv(\cone (\{X\} \cup \{Y\} ))\) are comonotone to one other.
\end{lemma}

\begin{proof}\label{prova lemma convex cone comono}
 In what follows all equalities and inequalities are in the $\Pro\otimes\Pro$-almost sure sense, that is, they hold for any pair $(\omega,\omega')$ lying in an event $\Omega_1\subseteq \Omega\times\Omega$ having total $\Pro\otimes\Pro$ measure.\footnote{$\Omega_1$ can be taken as the countable intersection of the events where the required inequalities (for any pairing of $X$, $Y$, $Y_n$, $Z$ and $W$) hold.}
 
 To see that $C_X$ is a cone, note that for any $Y \in C_X$ we have, by definition,
 \(
 \big(X(\omega) - X(\omega')\big) \times\left(Y(\omega)-Y(\omega') \right) \geq 0,
 \)
 for any $(\omega,\omega')\in \Omega_1.$ Hence, for any $\lambda \geq 0$ and {$(\omega,\omega')\in \Omega_1$},
 \[
 \big(X(\omega) - X(\omega')\big) \big(\lambda Y(\omega)- \lambda Y(\omega') \big)= \lambda \big(X(\omega) - X(\omega')\big) \big(Y (\omega)- Y (\omega') \big) \geq 0,
 \]
 yielding $\lambda Y \in C_X$. For convexity, let $Y,Z \in C_X$. Then, for $\lambda \in [0,1]$ we have that,
 \begin{align*}
  & \Big[ X(\omega) - X(\omega') \Big] \Big[ \big(\lambda Y(\omega) + (1-\lambda) Z(\omega)\big) - \big(\lambda Y(\omega') + (1-\lambda) Z(\omega')\big) \Big]
  \\ = &\lambda \left[ X(\omega) - X(\omega')\right] \left[ Y(\omega)-Y(\omega') \right] + (1-\lambda) \left[ X(\omega) - X(\omega')\right] \left[ Z(\omega)-Z(\omega') \right] \geq 0
 \end{align*}
 whenever $(\omega,\omega')\in \Omega_1$. {To see that $C_X$ is closed in the asserted sense, consider a convergent sequence $\{Y_n\}\subseteq C_X$ with $Y_n \to Y$ in probability. By standard facts of measure theory, there is a subsequence $\{Y_{n(k)}\}$ such that $Y_{n(k)}\to Y$ almost surely. Clearly, this yields that $Y$ is comonotone to $X$.}
 
 For the second claim
 , let $Z,W \in C_{X,Y} $. By definition we have
 \(Z = \gamma_1 (\lambda_1 X) + (1-\gamma_1)(\delta_1 Y)\)
 for some triplet $(\gamma_1, \lambda_1, \delta_1)$ with $0\leq \gamma_1\leq1$ and $0\leq \lambda_1,\delta_1$, and similarly
 \(W = \gamma_2 (\lambda_2 X) + (1-\gamma_2)(\delta_2 Y)\)
 for some triplet $(\gamma_2, \lambda_2, \delta_2)$ with $0\leq \gamma_2\leq1$ and $0\leq \lambda_2,\delta_2$. Then, for $(\omega,\omega')\in \Omega_1$, expanding the product
 \[
 \big(Z(\omega) - Z(\omega')\big)\big(W(\omega) - W(\omega')\big)
 \]
 yields a weighted sum whose terms are all non-negative. This completes the proof.
\end{proof}

\begin{remark}
Note that the set \(C \coloneqq \bigcap_{Y \in C_X} C_Y, \)
where $C_X$ and $C_Y$ are defined as in the proposition above, is a non-empty, closed, and convex set, such that all its elements are comonotone to one another. In particular, $\R \subseteq C$.
\end{remark}

\begin{coro}\label{coro como}
Let $A\subseteq\X$ be radially bounded at non-constants and closed for scalar addition. Furthermore, suppose both $A$ and $A^\complement$ are convex for comonotone pairs, i.e.\ $\lambda X + (1- \lambda) Y \in A$ for all $\lambda \in [0,1]$ whenever $X,Y \in A$ are comonotone, and similarly for $A^\complement$. Then $A$ is star-shaped and $\f$ a comonotone additive deviation measure. 
\end{coro}

\begin{proof}
First, notice that $A$ is star-shaped. Indeed, any $X \in A$ is comonotone to $0$, and by assumption $A$ is convex for this pair, i.e.\ $\lambda X\equiv \lambda X + (1-\lambda) 0 \in A$ for any $0\leq\lambda\leq1$. Furthermore, as $A$ is radially bounded at non-constants, it follows that $\cone (A^\complement) = ( \X\setminus \R) \cup \{0\}$ and so any cone with no constants that we may take is contained in $\cone(A^\complement)$.

Now let $X$ and $Y$ be a comonotone pair of non-constants. Note that any two members of the set $C_{X,Y} = \conv (\cone (\{X\} \cup \{Y\} ))$ are comonotone to one another (see \Cref{lemma convex cone comono}). Now, if we take any $Z,W \in C_{X,Y} \cap A$, as they are a comonotone pair, by assumption we have that $\lambda Z + (1-\lambda)W \in C_{X,Y} \cap A$. Hence, $ C_{X,Y} \cap A$ is a convex set. The same argument tells us that $ C_{X,Y} \cap A^\complement$ is also convex. Thus, by \Cref{linear}, we have that $\f(X+Y) = \f(X) + \f(Y)$. That $\f$ is a deviation measure follows from \Cref{radially}.
\end{proof}

\begin{remark}
If the conditions in the corollary above and in \Cref{linear} are imposed only on $A$ (and not necessarily on $A^\complement$), then we have in the proposition that $\f$ is convex on $C$, and in the corollary that $\f$ is convex for comonotone pairs. Similarly, if we only impose those conditions on $A^\complement$, then the resulting $\f$ is concave.
\end{remark}
%
%\begin{remark}
%Note that the assumptions on \Cref{coro como} above --- particularly radial boundedness --- imply that $\f (X)>0$ for any $X \in \X \setminus \{0\} $. Hence, such a set $A$ cannot yield a deviation measure, as it cannot fulfill the axiom of non-negativity. Notwithstanding, we can take $A+\R$ as the acceptance set --- which, due to \Cref{set oper-item4}, yields a \pl\ that satisfies $\D_{A+\R} (X) = \inf_{c \in \R} \f(X-c)$. Now let $X$ be non-constant, and notice that any constant is comonotone to $X$. From the \Cref{coro como}, we have that $ \inf_{c \in \R} \f(X-c) = \f(X) + \inf_{c \in \R} \f(-c) = \f(X) > 0$, whereas for constant $X$ it is clear that $\inf_{c \in \R} \f(X-c) = \inf_{c \in \R} \f(0) = 0 $. Therefore, $\D_{A+\R}$ is non-negative, and $A+\R$ is clearly stable under scalar addition. Consequently, in view of \Cref{radially}, \cref{radially item 2}, it holds that $\D_{A+\R}$ is translation insensitive, i.e.\ it is a deviation measure.
%\end{remark}
%
%\begin{coro}\label{A+R comonotonicidade}
%Let $A\subseteq\X$ be closed and radially bounded with $0\in A$. Suppose both $A$ and $A^\complement$ are convex for comonotone pairs. Then $\D_{A+\R}$ is a comonotone additive deviation measure. 
%\end{coro}

\subsection{Law invariance}

This subsection concerns law invariance. 

\begin{propo}\label{law}
If $\A \subseteq \X$ is law invariant then $\f$ is law invariant. 
\end{propo}
\begin{proof}
Let $X =_d Y \in \X$ and $m \in \R_+^*$. Clearly, one has $m^{-1}X=_d m^{-1}Y$ and thus, as $A$ is law invariant by assumption, the condition $m^{-1}X\in A$ holds if and only if it holds that $m^{-1}Y\in A$. This leads to
\[
\f (X) = \inf \left\{m \in \R_+^* \colon\, m^{-1}X \in A \right\} = \inf \left\{m \in \R^*_+ \colon\, m^{-1}Y \in A \right\} = \f (Y),
\]
thus proving the assertion.
\end{proof}

\begin{lemma}\label{law hull set}
Let $B\subseteq\X$. Then its \emph{law invariant hull} $\mathcal{L}_B \coloneqq \{X \in \X \colon\, \text{ $X =_d Y$, for some $Y \in B$} \} $ inherits from $B$ the attributes of stability under scalar addition, star-shapedness, absorbency, conicity and $\preceq_{\mathfrak{D}}$-monotonicity.
\end{lemma}
\begin{proof} 
If $B$ is stable under scalar addition, then taking any $Y\in\mathcal{L}_B$ and $c\in\R$ we see --- as, per definition, it holds that $Y =_d X$ for some $X\in B$ --- that $Y + c =_d X + c\in B$, that is $Y+c\in\mathcal{L}_B$.

Assume now that $B$ is a cone, and let $Y\in \mathcal{L}_B$ and $\lambda>0$. We have $Y=_d X$ for some $X\in B$, and, since $B$ is a cone, $\lambda X\in B$. But $\lambda Y =_d \lambda X$, and this is all we need to conclude that $\mathcal{L}_B$ is also a cone. A similar argument yields that $\mathcal{L}_B$ is star-shaped (resp., absorbing) whenever $B$ is.

%For symmetry, just note that $X =_d Y$ if and only if $-X =_d -Y$. Finally, $\preceq_{\mathfrak{D}}$-monotonicity is clear as the dispersive order of distributions is defined in terms of distributions alone.
\end{proof}

\begin{remark}
Not every property that seems plausibly heritable turns out to be so: take, for instance, radial boundedness of $B$. It seems reasonable --- since no random variable in $B$ can be scaled up indeterminately while remaining acceptable --- that the same should be true of $\mathcal{L}_B$. However, the counter\cref{law couter exemple} shows that this is false.
\end{remark}

We then have the following connection.

\begin{propo}
Let $B\subseteq\X$.
Then the equality
\[
\D_{\mathcal{L}_B} (X) = \inf_{Y \in\mathcal{L}_X} \D_{B} (Y)
\]
holds for all \(X\in\X.\)
\end{propo}

\begin{proof}
Let $X\in\X$ and $m\in\R_+^*$. Now, we have $X/m \in \mathcal{L}_B$ if and only if there exists an $Y\in\X$ such that $X/m =_d Y/m$ and $Y/m \in B$, if and only if there exists an $Y\in \mathcal{L}_X$ such that $Y\in mB$. Therefore, $\{m\in\R_+^*\colon X/m\in\mathcal{L}_B\} = \bigcup_{Y\in\mathcal{L}_X}\{m\in\R_+^*\colon\, Y\in mB\}$ and then
\begin{align*}
\D_{\mathcal{L}_B} (X)
%&= \inf \{m \in \R^*_+\colon\, X/m \in \mathcal{L}_B\} 
%\\ &= \inf \{m \in \R^*_+\colon\, X/m =_d Y/m \text{for some}Y/m \in B\} 
% &= \inf \{m \in \R^*_+\colon\, Y \in \mathcal{L}_X \text{for some} Y \in m B \}
&=\inf \bigcup\nolimits_{Y\in\mathcal{L}_X}\{m\in\R_+^*\colon\, Y\in mB\}
% \\ &= \inf \{m \in \R^*_+\colon\, \exists Y \in m B \cap \mathcal{L}_X \} 
\\ &= \inf\nolimits_{Y\in\mathcal{L}_X} \inf \{m \in \R^*_+\colon\, Y \in m B\}
% \\ &= \inf \left\{f_B (Y) \colon\, Y \in \mathcal{L}_X \right\}
\\ &= \inf\nolimits_{Y \in \mathcal{L}_X} \D_B(Y),
\end{align*}
as stated.
\end{proof}

\subsection{Monotonicity}

We now explore the fundamental relationship between set inclusion  and dominance of \pls\ and the effects of monotonicity with respect to a given partial order $\preceq$. Despite the fact that this kind of property is not studied much in the literature (both for gauges and deviations), it becomes crucial for decision making. Furthermore, some partial orders are specially suited for deviation measures, a fact which is illustrated in the following lemma.

\begin{lemma}
Let $A\subseteq\X$ be non-empty. If $A$ is $\preceq_{\mathfrak{D}}$-anti-monotone, then it is stable under scalar addition, star-shaped and law invariant. Hence, $\f$ is a law invariant deviation measure. 
\end{lemma}

\begin{proof}
Let $A$ be non-empty and assume it is $\preceq_{\mathfrak{D}}$-anti-monotone, and take any $X \in A$. Clearly $F^{-1}_c (u) - F^{-1}_c(v) = 0 $ for any $c \in \R$ and $0<v<u<1$. Hence, it is clear that $c \preceq_{\mathfrak{D}} X $ for any $c\in\R$, which entails $\R \subseteq A$. Moreover, notice that $F^{-1}_{X+c} (u) - F^{-1}_{X+c}(v) = F^{-1}_X (u) +c - F^{-1}_X(v) -c = F^{-1}_X (u) - F^{-1}_X(v) $ for any $c \in \R$ and $0<v<u<1$. Therefore, $ X +c \preceq_{\mathfrak{D}} X$ for all $c \in \R$, and due to anti-monotonicity of $A$, we get $X + c \in A$ (this holds for all $X \in A$ and $c \in \R$).
Furthermore, $A$ is star-shaped: indeed, given $X \in A$ we have $F^{-1}_{X} (u) - F^{-1}_{X}(v) \geq \lambda \big(F^{-1}_{X} (u) - F^{-1}_{X}(v)\big) = F^{-1}_{\lambda X} (u) - F^{-1}_{\lambda X}(v), $ for any $\lambda \in [0,1]$ and $0 <v<u<1$. Hence, $ \lambda X \preceq_{\mathfrak{D}} X$ for any $\lambda \in [0,1]$, from which star-shapedness of $A$ follows. Additionally, anti-monotonicity w.r.t.\ $\preceq_{\mathfrak{D}}$ clearly implies that $A$ is law invariant, as if $Y$ and $X$ follow the same distribution it is obvious that $Y \preceq_{\mathfrak{D}} X \preceq_{\mathfrak{D}} Y$. The assertion about $\f$ is a direct consequence of \Cref{radially} and \Cref{law}.
\end{proof}

Intuitively, the larger the acceptance set, the more permissive (w.r.t risk taking) the agent becomes. The following results should thus come at no surprise.

\begin{lemma} \label{lemma} 
Let $ A , \B \subseteq \X$ be non-empty and $\mathfrak{B}$ be a family of sets.
Then the following holds:
\begin{enumerate}[label = (\roman*)]
\item\label{lemma-item set order} If $ \A \subseteq \B$, then $ \D_\A (X) \geq \D_{\B} (X)$, for all $X \in \X$.

\item\label{set oper-item1} $\D_{ \bigcup \{A: A \in \mathfrak{B} \}} (X) = \inf_{ \{A: A \in \mathfrak{B} \}}\D_\A (X)$.

\item \label{lemma-item5} If $A$ and $B$ are star-shaped, then $\D_{\A \cap \B} (X) = \max (\f (X) , \D_\B (X) )$

\item\label{set oper-item scale} $ \D_{\lambda \A} (X) = \lambda^{-1}{\f (X)}$, for every $\lambda \in \R_+^*$.

%\item\label{lemma-item-sym} $\f (-X)= \D_{-\A}(X)$ for all $X \in \X$; in particular, if $\A $ is symmetric, so is $\f$.

\end{enumerate}
\end{lemma}

\begin{proof}
\Cref{lemma-item set order,lemma-item5} can be found in \cite{aliprantis06} Lemma 5.49. The remaining assertions are obvious.
\end{proof}

\begin{coro}\label{lrdominance} If $\Acc{1}{\LR} \subseteq A $ then $\f$ is lower range dominated.
\end{coro}

\begin{remark}
A natural way to force lower-range dominance is by taking an acceptance set of the form $A = B \cup \Acc{1}{\LR}$, where $B\subseteq\X$ is a given set of acceptable positions. This yields $\f = \min (\D_{B} , \D_{\Acc{1}{\LR}})$. However, while the union operation preserves properties like stability under scalar addition, star-shapedness, law invariance and radial boundedness at non-constants, it is possible that convexity may be lost.
\end{remark}

\begin{propo} \label{monotonicidades}
Let $\preceq$ be a partial order that is stable under positive scalar multiplication,\footnote{That is to say, it holds that $Y \preceq X$ if and only if $ \lambda Y \preceq \lambda X$ for all $\lambda \in \R_+$.} and let $\A \subseteq \X$. Then, we have the following:
\begin{enumerate}[label = (\roman*)]
\item If $\A$ is monotone with respect to $\preceq$, then $\f$ is anti-monotone with respect to $\preceq$.

\item If $\A$ is anti-monotone with respect to $\preceq$, then $\f$ is monotone with respect $\preceq$.
\end{enumerate}
\end{propo}

\begin{proof}
For the first item, let $X \preceq Y$. If $ m \in \R_+^*$ is such that $m^{-1}X \in \A$, then $ m^{-1}Y \in \A$, as $\A$ is monotone. Thus, $ \{m \in \R_+^* \colon\, m^{-1}X \in \A \} \subseteq \{m \in \R_+^* \colon\, m^{-1}Y \in \A \} $ and hence $\f (X) \geq \f (Y)$.

Similarly, for item (ii) let $Y \preceq X$. If $ m \in \R_+^*$ is such that $m^{-1}X \in \A$, then $ m^{-1}Y \in \A$, as $\A$ is anti-monotone. Thus, $ \{m \in \R_+^* \colon\, m^{-1}X \in \A \} \subseteq \{m \in \R_+^* \colon\, m^{-1}Y \in \A \} $ and hence $\f (X) \geq \f (Y)$.
\end{proof}

\begin{propo}
Let $\varnothing\neq A\subseteq\X$ be $\preceq_{\mathfrak{D}}$-anti-monotone. Then
\begin{enumerate}[label = (\roman*)]
\item If $(A, \preceq_{\mathfrak{D}})$ has a greatest element $X$, then $A$ is stable under convex combinations of comonotone pairs and radially bounded at non-constants. Furthermore, $\f$ admits the following representations:
\begin{align*}
\f(Y) &= \inf \{m \in \R_+^*\colon\, F^{-1}_{Y} (u) - F^{-1}_{Y}(v) \leq m\big(F^{-1}_{X} (u) - F^{-1}_{X}(v)\big) ,\, \forall\; 0<v<u<1\} 
\\ &= \sup\{m \in \R_+^*\colon\, F^{-1}_{Y} (u) - F^{-1}_{Y}(v) > m\big(F^{-1}_{X} (u) - F^{-1}_{X}(v)\big) , \text{ for some } 0<v<u<1 \} 
\\& = \inf \{m \in \R_+^*\colon\, Y \preceq_{\mathfrak{D}} m X \}. 
\end{align*}

\item If $(A^\complement, \preceq_{\mathfrak{D}})$ has a least element $X$, then $A^\complement$ is stable under convex combinations of comonotone pairs. Furthermore, $\f$ admits the following representation:
\begin{align*}
\f(Y) &= \inf \{m \in \R_+^*\colon\, F^{-1}_{Y} (u) - F^{-1}_{Y}(v) < m\big(F^{-1}_{X} (u) - F^{-1}_{X}(v)\big) , \text{ for some } 0<v<u<1 \} 
\\
& = \sup\{m \in \R_+^*\colon\, F^{-1}_{Y} (u) - F^{-1}_{Y}(v) \geq m\big(F^{-1}_{X} (u) - F^{-1}_{X}(v)\big) ,\, \forall\; 0<v<u<1 \} .
\\& = \sup \{m \in \R_+^*\colon\, mX \preceq_{\mathfrak{D}} Y \}. 
\end{align*}
\end{enumerate}
\end{propo}

\begin{proof}
Before proceeding, notice that $A$ is necessarily star-shaped.

For item (i), let $X$ be the greatest element of $A$. First, note that the quantile function is comonotone additive, in the sense that $F^{-1}_{Y+Z} = F^{-1}_Y + F^{-1}_Z$ whenever $(Y,Z)$ is a comonotone pair --- see Lemma~4.90 in \cite{follmer02}. Hence, for any $W$ that is a convex combination of some comonotone pair $Y,Z \in A$, it follows that
\begin{align*}
F^{-1}_{W} (u) - F^{-1}_{W}(v)
&= F^{-1}_{\lambda Z + (1- \lambda )Y} (u) - F^{-1}_{\lambda Z + (1- \lambda )Y}(v) 
\\ &= \lambda F^{-1}_{Z} + (1- \lambda )F^{-1}_{Y} (u) - \lambda F^{-1}_{Z} + (1- \lambda )F^{-1}_{Y}(v)
\\ &= \lambda \big(F^{-1}_{Z} (u) -F^{-1}_{Z}(v)\big) + (1-\lambda) \big(F^{-1}_{Y} (u) - F^{-1}_{Y}(v)\big)
\\ &\leq \max \left(F^{-1}_{Z} (u) - F^{-1}_{Z}(v), F^{-1}_{Y} (u) - F^{-1}_{Y}(v) \right)
\\ &\leq F^{-1}_{X} (u) - F^{-1}_{X}(v),
\end{align*}
for all $0<v<u<1$, which shows that $W \in A$.

To see that $A$ is radially bounded at non-constants, note that one has $F^{-1}_{Y} (u) - F^{-1}_{Y}(v) = 0 $ for all $0<v<u<1$ if and only if $Y$ is constant. Hence, for a non-constant $Y$, there is some $u$ and $v$ with $u>v$ such that $c \coloneqq F^{-1}_{Y} (u) - F^{-1}_{Y}(v) > 0 $. Also, we have that $\lambda c =F^{-1}_{\lambda Y} (u) - F^{-1}_{\lambda Y}(v) $ for any $\lambda >0 $. Therefore, as $ k \coloneqq F^{-1}_{X} (u) - F^{-1}_{X}(v) \geq F^{-1}_{Y} (u) - F^{-1}_{Y}(v)$, it is obvious that one can find a $\gamma$ such that for any $\lambda \geq \gamma$ the inequality $\lambda c > k$ holds. This implies that $ \lambda Y \preceq X $ never holds, and hence --- as $X$ is the greatest element of $A$ --- we must have $\lambda Y \notin A$. As $Y \in A$ was arbitrary, it follows that $A$ is radially bounded at non-constants.

For the stated representations, note that $Y \in A$ if and only if $Y \preceq X$. Therefore, the following holds
\begin{align*}
\f (Y) &= \inf \{m \in \R_+^* \colon\, m^{-1} Y \in A \} 
\\ &= \inf \{m \in \R_+^* \colon\, Y \preceq m X \} 
\\ &= \inf \{m \in \R_+^* \colon\, F^{-1}_{Y} (u) - F^{-1}_{Y}(v) \leq m\big(F^{-1}_{X} (u) - F^{-1}_{X}(v)\big),\, \forall \; 0<v<u<1 \} .
\end{align*}
Furthermore, remember that $0 \in A $ and that, if $A$ is stable under convex combinations of comonotone pairs, then $A$ is star-shaped (see \Cref{coro como}). Hence we have, by \Cref{coro cogauge}, that $\f = \cog_{A^\complement}$ and so
\begin{align*}
\f (Y) &= \cog_{A^\complement}(Y)
\\ &= \sup \{m \in \R_+^* \colon\, m^{-1} Y \in A^\complement \} 
\\ &= \sup \{m \in \R_+^* \colon\, m^{-1} Y \notin A \} 
\\ &= \sup \{m \in \R_+^* \colon\, Y \preceq m X \text{ does not hold} \} 
\\ &= \sup \{m \in \R_+^* \colon\, F^{-1}_{Y} (u) - F^{-1}_{Y}(v) > m\big(F^{-1}_{X} (u) - F^{-1}_{X}(v)\big), \text{ for some } 0<v<u<1\} .
\end{align*}

For the second item, let $X$ be the least element of $A^\complement$. First, we shall show that $A^\complement$ is monotone with respect to the dispersive order of distributions: let $Y \in A^\complement$ and $Y \preceq Z$. Suppose, by contradiction, that $Z \in A$. Then, as $A$ is anti-monotone, we should have $Y \in A$, an absurd. Hence, $Z \in A^\complement$.
Now, notice that for any $W$ that is a convex combination of some comonotone pair $Y,Z \in A^\complement$, i.e.,\ $W = \lambda Z + (1-\lambda)Y$ for some $\lambda$ in the unit interval, the following holds for all $0<v<u<1$: 
\begin{align*}
F^{-1}_{W} (u) - F^{-1}_{W}(v)
&= F^{-1}_{\lambda Z + (1- \lambda )Y} (u) - F^{-1}_{\lambda Z + (1- \lambda )Y}(v) 
\\ &= \lambda F^{-1}_{Z} + (1- \lambda )F^{-1}_{Y} (u) - \lambda F^{-1}_{Z} + (1- \lambda )F^{-1}_{Y}(v)
\\ &= \lambda (F^{-1}_{Z} (u) -F^{-1}_{Z}(v)) + (1-\lambda) (F^{-1}_{Y} (u) - F^{-1}_{Y}(v))
\\ &\geq \min \big(F^{-1}_{Z} (u) - F^{-1}_{Z}(v), F^{-1}_{Y} (u) - F^{-1}_{Y}(v) \big)
\\ &\geq F^{-1}_{X} (u) - F^{-1}_{X}(v).
\end{align*}
Therefore, as $A^\complement$ is monotone w.r.t.\ $\preceq_{\mathfrak{D}}$, we have $W \in A^\complement$. 

Finally, for the stated representations note that $Y \in A^\complement$ if and only if $X \preceq Y$. Therefore, the following holds,
\begin{align*}
\f (Y) &= \inf \{m \in \R_+^* \colon\, m^{-1} Y \in A \} 
\\ &= \inf \{m \in \R_+^* \colon\, m^{-1} Y \notin A^\complement \} 
\\ &= \inf \{m \in \R_+^* \colon\, mX \preceq Y \text{ does not holds} \} 
\\ &= \inf \{m \in \R_+^* \colon\, F^{-1}_{Y} (u) - F^{-1}_{Y}(v) < m\big(F^{-1}_{X} (u) - F^{-1}_{X}(v)\big), \text{ for some } 0<v<u<1\} .
\end{align*}
For the second representation, as $A$ is star-shaped, we are once again allowed to summon the cogauge in order to obtain
\begin{align*}
\f (Y) &= \cog_{A^\complement}(Y)
\\ &= \sup \{m \in \R_+^* \colon\, m^{-1} Y \in A^\complement \} 
\\ &= \sup \{m \in \R_+^* \colon\, mX \preceq Y\} 
\\ &= \sup \{m \in \R_+^* \colon\, F^{-1}_{Y} (u) - F^{-1}_{Y}(v) \geq m(F^{-1}_{X} (u) - F^{-1}_{X}(v)), \forall \, 0<v<u<1\} .
\end{align*}
This completes the proof.
\end{proof}

\subsection{Dual Representation}

We now define and explore a very important concept regarding duality in convex analysis, namely the \emph{polar} of a set. This concept is of particular interest to us as it is the subgradient of the \pl\ at $0$.

\begin{defi} For a dual pair $\langle \X, \X' \rangle$, the \defin{polar $A^\odot$} of a non-empty set $A\subseteq\X$ is defined through
\[
A^{\odot} \coloneqq \{X' \in \X'\colon\, {\sup\nolimits_{X\in A}\left\langle X, X' \right\rangle \leq 1} \},
\]
and the \defin{bipolar} of $A$ is the set given by
\[
A^{\odot\odot}\coloneqq \big\{X\in\X\colon\, \sup\nolimits_{X'\in A^\odot}\langle X,X'\rangle\leq1 \big\}.
\]
\end{defi}

\begin{remark}
Notice that the bipolar is always defined with the dual pair $\langle \X, \X'\rangle$ in mind, which forces the inclusion $A^{\odot\odot}\subseteq \X$. If instead one had the bidual $\X''$ in mind (or, which is the same, the dual pair $\langle \X', \X''\rangle$), it would then be natural to define $(A^{\odot})^\odot\coloneqq \big\{X'' \in\X'' \colon\, \sup\nolimits_{X'\in A^\odot}\langle X',X''\rangle\leq1 \big\}$. In this case, however, unfortunately one may have
$A^{\odot\odot}\neq(A^\odot)^\odot$. This is a detail that is frequently overlooked in the literature, although it has important consequences: for instance, see the \emph{Bipolar Theorem} (\cref{bipolar-theorem} in \Cref{polar operations}), and also \Cref{bipolar-counterexample}.
\end{remark}

\begin{lemma} \label{polar operations} Given a dual pair $\langle \X,\X' \rangle$, let $A,B, \{A_i\}_{i \in I}$ be subsets of $\X$:
\begin{enumerate}[label = (\roman*)]
\item\label{polar operations item inclusion} If $A \subseteq B$, then $ B^\odot \subseteq A^\odot$.
\item\label{polar item 2} $(\lambda A)^\odot = \lambda^{-1} A^\odot$ for each $\lambda\neq0$.
\item\label{polar operations item union intersection} $ \cap A_i^\odot = (\cup A_i)^\odot$.
\item\label{polar item 4} $A^\odot$ is nonempty, convex, weakly$^*$-closed and contains 0.

\item\label{polar operations item 3} If $ A$ is absorbing, then $A^\odot$ is weak{ly}*-bounded, i.e.,\ the set $\{\left\langle XX' \right\rangle \colon\, X \in A \} $ is bounded in $\R$, for every $X' \in \X'$.

\item\label{bipolar-theorem} The bipolar $A^{\odot \odot}$ is the convex, weak-closed hull of $ A \cup \{0\} $.

\item \label{polar do cone}
If $A$ is a cone, then $A^\odot = \{X'\in\X'\colon\, \langle X,X'\rangle \leq 0, \forall \; X \in A \} $.
\item \label{lemma radial}If $A$ is star-shaped and stable under scalar addition, then $\left\langle 1 , X' \right\rangle = 0$ for all $X' \in A^\odot$. 
\end{enumerate}
\end{lemma}

\begin{proof}
For \cref{polar operations item inclusion,polar item 2,polar operations item union intersection,polar item 4,polar operations item 3,bipolar-theorem}, see Lemma 5.102 and Theorem 5.103 of \cite{aliprantis06}. \Cref{polar do cone} follows from an argument similar to the proof that $B_0 = B_0^*$ in \Cref{polar alter}.
For \cref{lemma radial}, let $X'\in A^\odot$. Then --- as $\R\subseteq A$ and $A+\R\subseteq A$ by assumption --- we have, for any $X\in A$ and $c\in\R$,
\[
\langle X, X'\rangle + c \langle1,X'\rangle = \langle X + c, X'\rangle \leq 1
\]
and, as $c$ is arbitrary, it is necessarily true that $\left\langle 1 , X' \right\rangle = 0$.
\end{proof}

\begin{remark}
\Cref{bipolar-theorem} above is the famous \emph{Bipolar Theorem}, which states, in other words, that if $A$ is closed, convex and contains zero, then $A=A^{\odot\odot}$. It is important to have in mind that $A^{\odot\odot} \subseteq \X$ \emph{by definition}. The (counter)\cref{bipolar-counterexample} provides reasoning for the bipolar to be defined in $\X$ and not in $\X''$.
\end{remark}

A well-known result in convex analysis is the duality associating the \pl\ of a set $A$ with the support function of its polar
\(
h_{A^\odot} (X) \coloneqq \sup\nolimits_{X' \in A^{\odot}} \left\langle X, X' \right\rangle,
\)
$X\in\X$. This is related to the convex biconjugate of the Fenchel-Moreau Theorem (when $\X$ is a locally convex topological space) via the conjugate and biconjugate functions (the latter is also called the \emph{penalty function} in the jargon of convex risk measures). If the \pl\ is a proper, convex and weakly lower-semicontinous functional, then the penalty is precisely the characteristic function\footnote{The characteristic function of the polar assumes $0$ if $X \in A^\odot$ and $\infty$ otherwise.} of the polar. Below we present this duality result for topological vector spaces, without relying on the Frenchel-Moreau Theorem. 

\begin{defi} The \defin{support function} $h_{A^\odot}\colon\X \rightarrow \ra$ on the polar $A^\odot$ is defined, for $X\in\X$, as
\[
h_{A^\odot} (X) \coloneqq \sup \{\left\langle X, X' \right\rangle\colon\, X' \in A^{\odot} \} .
\]
\end{defi}

\begin{propo}[Dual representation]\label{dual-representation}
Let $A $ be a closed, convex set such that $0 \in A$. Then we have the identity
\begin{equation}\label{eq:dual representation}
\f (X) = h_{A^\odot} (X)
\end{equation}
for all $X\in \X$.
\end{propo}
\begin{proof}
For simplicity, let us write $h\coloneqq h_{A^\odot}$. First of all, notice that $ h\colon \X \rightarrow \ra$ is a lower-semicontinuous, sub-linear function. Then, by \Cref{lemma-item4}, we have
\[ h = \D_{\Acc{1}{h}}.\]
Therefore, it is enough to show that $\Acc{1}{h} = A$. Note that 
\begin{align*}
\Acc{1}{h} &= \{X \in \X \colon\, h(X) \leq 1 \} 
\\ &= \big\{X \in \X \colon\, \sup\nolimits_{X' \in A^\odot} \left\langle X,X' \right\rangle \leq 1 \big\}
\\ &= \big\{X \in \X \colon\, \left\langle X,X' \right\rangle \leq 1 \text{ for all } \X' \in A^\odot \big\}
\\ &= A^{\odot \odot}.
\end{align*}
Lastly, the Bipolar Theorem (\cref{bipolar-theorem} in \Cref{polar operations}) entails $A = A^{\odot \odot}$. Hence,
\[ h = \D_{\Acc{1}{h}} = \D_{A^{\odot \odot}} = \f\]
as claimed.
\end{proof}

\begin{remark}
The equality in the proposition above holds even if $A$ is empty, as in this case $\f \equiv \infty$ and $A^\odot = \X'$, so $ h_{A^\odot} \equiv \infty$. It is also interesting to remember that if $A = \X$, then $\f \equiv 0$, $A^\odot = \{0\} $ and $ h_{A^\odot} \equiv 0$. Furthermore, note that ${A^\odot}=\partial\f(0)\coloneqq $ ``the set of sub-gradients of $\f$ at $0$''.
\end{remark}

\begin{remark}
By the Bipolar Theorem, for a closed, star-shaped set $A$, we have that $A^{\odot \odot} = \conv A$. Additionally, \Cref{dual-representation} above tells us that $h_{A^\odot} = \D_{A^{\odot \odot}} = \D_{\conv A}$. However, by \Cref{conv envelop}, as $A$ is closed, $\D_{\conv A} = \conv \f$. Therefore, we have the following representation for the support function $h_{A^\odot}$ in terms of the convex envelope of $\f$:
\begin{equation}\label{dual conv envelop}
h_{A^\odot}(X) = \conv \f(X),\qquad X\in\X.
\end{equation}
Note that, as the convex hull operation preserves many properties, such as stability under scalar addition, radial boundedness at non-constants and law invariance, and since $A \subseteq \conv A$, we have that $h_{A^\odot}$ is translation insensible, non-negative, law invariant and lower range dominated if $A$ satisfies each of those properties respectively.
\end{remark}

An important result in the literature of risk and deviation measures is the \emph{dual representation} for convex deviation measures, \Cref{dual} below. We highlight that the result in \Cref{dual-representation} is an intermediary step in the proof of the following theorem.

\begin{theorem}[\cite{rockafellar06}, Theorem 1]\label{dual}
A given functional $D\colon L^2\to{\R_+}\cup\{+\infty\}$ is a {lower-semicontinuous generalized deviation measure} if and only if it it has a representation of the form
\begin{equation}\label{eq:dual-rockafellar06}
D(X) = \E X - \inf_{Q \in \mathcal{Q}} \E[XQ],\quad X\in L^2
\end{equation}
in terms of a convex envelope $\mathcal{Q} \subseteq L^2$ satisfying the following:
\begin{enumerate}[noitemsep, topsep=0pt, label = (\emph{Q\arabic*})]
\item $\mathcal{Q}$ is non-empty, closed and convex;
\item for each non-constant $X$ there is a $Q\in\mathcal{Q}$ for which $\E(XQ)<\E X$;
\item $\E Q = 1$ for all $Q\in\mathcal{Q}$.
\end{enumerate}
Additionally, the set $\mathcal{Q}$ above is uniquely determined by $D$ through
\[
\mathcal{Q} = \{Q \in L^2 \colon D(X) \geq \E X - \E[XQ] \text{ for all } X \} ,
\]
and the finiteness of $D$ is equivalent to boundedness of $\mathcal{Q}$. Furthermore, $D$ is lower-range dominated if and only if $\mathcal{Q}$ has the additional property that
\begin{enumerate}[noitemsep, topsep=0pt, label = (\emph{Q\arabic*})]
\setcounter{enumi}{3}
\item $Q\geq 0$ for all $Q\in\mathcal{Q}$.
\end{enumerate}
\end{theorem}

With regard to our framework, we have the following correspondences for the dual representation in the generalized and law invariant cases.

\begin{coro}
Let $A\subseteq\X$. Suppose $A$ is convex, radially bounded, stable under scalar addition and contains the origin. Then $\f$ is a generalized deviation measure, and admits the dual representation
\[
\f (X) = \E X - \inf_{Q\in \mathcal{Q}} \E[XQ] = \sup_{X'\in A^\odot} \left\langle X, X' \right\rangle = h_{A^\odot} (X) ,\qquad X\in\X,
\]
where $\mathcal{Q} = 1 - A^\odot$. Furthermore, if $\Acc{1}{\LR} \subseteq A$, then $X^\prime\leq 1$ for any $X^\prime\in A^\odot$.
\end{coro}

\begin{propo}\label{quantil rep}
Assume $(\Omega,\mathfrak{F},\Pro)$ is 
an atomless probability space, and put $\X \coloneqq L^p$ ($p\in[1,\infty]$). Let, moreover, $B$ denote a law invariant, closed, radially bounded, convex subset of $\X$ containing the origin, and define $A \coloneqq B + \R$. Then $\f $ is a law invariant, lower semicontinuous generalized deviation measure, and the following representation holds, for all $X\in\X$:
\[
\f (X)
= \sup_{X' \in A^\odot} \int_0^1 F^{-1}_X (t) F^{-1}_{X'} (t) \mathrm{d}t\\
= \sup_{\psi \in \Lambda} \int_0^1 \psi (t) F^{-1}_X (t) \mathrm{d}t
= \sup_{g \in G} \int_0^1 g(t) F^{-1}_X (\mathrm{d}t),
\]
where $\Lambda $ is a collection of nondecreasing functions $\psi \in \Lq[0,1] $ such that $\int_0^1 \psi (t)\, \dd t =0$, and where
$G$ is a collection of positive concave functions $g\colon [0,1] \rightarrow \R$. If in addition $B^\complement$ is convex for comonotone pairs, then $\f$ is also comonotone additive, and for each $X\in\X$ the supremma in the above representations is attained for some $X^\prime\in A^\odot$, $\psi\in\Lambda$, $g\in G$.
\end{propo}

\begin{proof}
Let $A$ be a law invariant, closed, radially bounded at non-constants, stable under scalar addition, convex set containing the origin. Then,  this yields that $\f$ is a law invariant, lower semi continuous generalized deviation measure. The stated representations follow from Propositions 2.1 and 2.2 of \cite{grechuk09}. Also, Proposition 2.4 of the same paper yields, under comonotonic additivity --- which is given by the convexity for comonotonic pairs of $A^\complement$ and \Cref{coro como} --- that
\(
\f (X) = \int_0^1 g(t)\, F^{-1}_X (\mathrm{d}t),
\)
for some positive concave function $g\colon [0,1] \rightarrow \R$. 
\end{proof}

\begin{remark}
By taking $B = \Acc{1}{\varepsilon}$ for a law invariant measure of error $\varepsilon$, one obtains a set $B$ that fulfills the requirements from the above proposition. To ensure comonotonicity, one can take $B$ of the form $B=\Acc{1}{\varepsilon}$ as above, with the additional requirement that the error measure $\varepsilon$ be comonotone additive.
\end{remark}

\section{Acceptance sets for deviation measures}\label{accept}

So far we have (mostly) focused on the scenario where an acceptance set $A\subseteq\X$ is given, and studied the relations existing between attributes of this set and the associated features of its \pl, especially how the former manifest on the latter. Now, a special case occurs when the acceptance set itself is already induced by a given, specified \emph{a priori} deviation measure. Remember that under the mild requirement that $A$ is closed and star-shaped we have $A = \Acc{1}{\f}$ (as ensured by \Cref{lemma-item4} ). Additionally, \cref{lemma-item6} in \Cref{lemma2} tells us that a positive homogeneous function $D$ coincides with the \pl\ of $\Acc{1}{D}$ (where the requirement that the underlying set is absorbing may be dropped when $+\infty\in\mathrm{range}D$). 

The crucial fact explored in this section is that we actually have a \emph{two-way correspondence} between attributes of the functional and the properties of the associated acceptance set. In particular, a lower-semicontinuous, convex deviation measure yields an acceptance set which is stable under scalar addition, convex, closed and radially bounded at non-constants.
%
%\begin{propo}\label{some-ph-equalities}
%Let $\left\{A(k) \colon\, k \in \R_+^*\right\}$ {be a positive homogeneous family of subsets of $\X$}. Then, for each pair of positive real numbers $k$ and $\lambda$, it holds that
%\begin{equation} \label{eq ph family}
%k\, \D_{A(k)} = \lambda\, \D_{A(\lambda)} = \inf \{m \in \R_+^* \colon\, X \in \A(m) \} .
%\end{equation}
%\end{propo}
%
%\begin{proof}
%For the first equality, let $ k, \lambda \in \R_+^*$. Due to positive homogeneity, we have
%$ k^{-1}{\lambda} A(k) = A(\lambda)$. Therefore,
%\begin{align*}
%k \D_{A(k)} (X) &= \lambda \inf \left\{mk\lambda^{-1} \in \R_+^* \colon\, X \in mA(k) \right\}
%\\ &= \lambda \inf \left\{m \in \R_+^* \colon\, X \in mk^{-1}\lambda A(k) \right\}
%\\ &= \lambda \inf \left\{m \in \R_+^* \colon\, X \in mA(\lambda) \right\}
%\\ &= \lambda \D_{A(\lambda)} (X).
%\end{align*}
%The second equality follows trivially by setting $\lambda = 1$ above and noticing that $mA(1) = A(m)$.
%\end{proof}
%
%\begin{remark}
%The representation $\inf \{m \in \R_+ \colon\, X \in \A(m) \} $ appearing in \cref{eq ph family} was studied in the context of risk measures in \cite{drapeau13}, under some extra conditions on the family $\{A(k)\} $.
%\end{remark}
For the theorem below, recall that $\Acc{k}{D} \coloneqq \{X\in\X\colon\,D(X)\leq k\}$. {The following theorem provides a characterization for acceptance sets generated by deviation measures, i.e.,\ sub-level sets corresponding to non-negative, translation insensitive functionals on $\X$.
These results are new in the literature, and can be seen as reciprocals for the results studied in the previous sections.}

\begin{theorem} \label{deviation a.set}
Let $ D,D'\colon\X \rightarrow \R \cup \{\infty\} $ be positive homogeneous functionals. Then we have the following, for all positive real $k$,
\begin{enumerate}[label = (\roman*)]
\item\label{deviation a.set-item1}  $ \Acc{\lambda}{D}$  is star-shaped set for all $ \lambda \in \R_+$. Moreover, if $D$ does not assume negative values, then the following string of equalities holds, for all $X\in\X$:
\begin{equation*}
D(X) = \D_{\Acc{1}{D}}(X)
%&= k \inf \left\{m \in \R_+^*\colon\, m^{-1}X \in \Acc{k}{D}\right\}\\
= k \D_{\Acc{k}{D}} (X)
.
\end{equation*}

\item If $ D$ is finite, then $\Acc{k}{D}$ is absorbing.

\item If $ D$ is translation insensitive, then $ \Acc{k}{D} + \R = \Acc{k}{D}$.

\item If $ D$ is non-negative, then $\Acc{k}{D}$ is radially bounded at non-constants and $\R \subseteq \Acc{k}{D}$.

\item\label{deviation a.set-item convex} If $D$ is a convex functional, then $ \Acc{k}{D}$ is a convex set.

\item\label{deviation a.set-item concave} If $D$ is a concave functional, then $ (\Acc{k}{D})^\complement$ is a convex set.

\item If $D$ is law invariant, then so is $\Acc{k}{D}$.

\item\label{deviation a.set-item8} If $D \leq D' $, then $\Acc{k}{D'} \subseteq \Acc{k}{D}$. In particular, if $D$ is lower-range dominated then the inclusion $\Acc{k}{\LR} \subseteq \Acc{k}{D}$ holds.

%\item If $f $ is symmetric, then so is $ \Acc{k}{f}$.

\item If $D(X) > 0 $ for all $X\in\X$, then $ \Acc{k}{D}$ is radially bounded.

\item\label{deviation a.set-item additivity} If $D $ respects $D(X+Y) = D(X) + D(Y)$ for $X,Y$ in some convex cone $C$, then $ \Acc{k}{D} \cap C$ and $(\Acc{k}{D})^\complement \cap C$ are convex sets. In particular, if $D$ is comonotone additive, then both $\Acc{k}{D}$ and its complement are stable under convex combinations of comonotone pairs in $\X$.

\item\label{deviation a.set-item12} If $D$ is lower-semicontinuous, then $\Acc{k}{D}$ is closed.

\item\label{deviation a.set-item strong shaped} If $D$ is continuous, then $\Acc{k}{D}$ is strongly star-shaped.

\item\label{deviation a.set-item monotonicity} If $D$ is monotone, then $\Acc{k}{D}$ is anti-monotone and $\Acc{k}{-D}$ is monotone.
\end{enumerate}
\end{theorem}

\begin{proof}
For item (i), star-shapedness of each $\Acc{\lambda}{D}$ is clear as if $ D (X) \leq k$, then $ \lambda D (X) \leq k$, for any $ \lambda \in [0,1]$. Also, note that (by positive homogeneity of $D$)
\[
\lambda \Acc{k}{D} = \left\{\lambda X \in \X \colon\, D (X) \leq k \right\} = \left\{X \in \X \colon\, D(X) \leq \lambda k \right\} = \Acc{\lambda k}{D}.
\]
 It remains to prove that $D = \D_{\Acc{1}{D}}$, as the remaining equalities will follow. Now, with $A = \Acc{1}{D}$, we have (again by positive homogeneity of $D$)
\begin{align*}
\f(X) %&= \inf\Big\{m\in\R_+^*\colon\, m^{-1}X \in \{Z\colon D(Z)\leq 1\} \Big\} \\
&= \inf\Big\{m\in\R_+^*\colon\, D(X)\leq m\Big\} = D(X).
\end{align*}

Item (ii) is clear, as if $D$ is a positive homogeneous finite function and $k>0 $ then, for any $X \in \X$ such that $D(X) >0 $ one has $D \big(k X / D(X)\big) = k $. Therefore, we have $t X \in \Acc{k}{D}$ for any $0\leq t \leq \delta_X \coloneqq k / D(X)$. Of course, if $D(X) \leq 0$ then there is nothing to prove, {as in this case we have $D(X)\leq k$, that is, $X\in \Acc{k}{D}$.}

% For item (iii), let $X \in \Acc{k}{f}$. It is enough to show that $X + c \in \Acc{k}{f}$ for any constant $c$. But clearly, \[ f(X+c) = f(X) \leq k,\]
% hence, the claim follows.

For item (iii), let $Y\in \Acc{k}{D} + \R$, that is, $Y = X + c$ for some $X\in \Acc{k}{D}$ (meaning $D(X)\leq k$) and some $c\in\R$. Then $D(Y) = D(X+c) = D(X) \leq k$ as $D$ is translation insensitive. This yields $\Acc{k}{D} + \R \subseteq \Acc{k}{D}$. The reverse inclusion holds trivially.

Item (iv) follows from the fact that, for any non-constant $X$, we have $D(X) > 0$ (by assumption). Hence, by positive homogeneity of $D$, there is some $\delta_X \coloneqq k/D(X) > 0 $ such that $ D(m X) > k$ for all $m>\delta_X$. Furthermore, as $D(c) = 0 < k$ for any $c \in \R$ it follows that $\R \subseteq \Acc{k}{D}$.

For item (v), let $X,Y \in \Acc{k}{D}$ and let $Z$ be any convex combination of $X$ and $Y$. It follows from the convexity of $D$ that $D(Z) \leq \max(D(X),D(Y)) \leq k$, hence the claim holds.

For item (vi), let $B = (\Acc{k}{D})^\complement$, $X,Y \in B$ and assume $Z$ is any convex combination of $X$ and $Y$. It follows from the concavity of $D$ that $D(Z) \geq \min(D(X),D(Y)) > k$, hence the claim holds.

Regarding item (vii), let $X \in \Acc{k}{D}$ and assume $Y =_dX$. Then, due to law invariance of $D$, we have
\(D(Y) = D(X) \leq k,\), that is $Y \in \Acc{k}{D}$.

For item (viii), let $X \in \Acc{k}{D'}$. Clearly, the claim holds, as
\(D(X) \leq D'(X) \leq k.\) The particular case for when $D$ is lower-range dominated is obvious from the definition.

%To prove item (ix) simply note that if $X \in \Acc{k}{f}$, then --- due to the symmetry of $f$ --- it holds that \(f(-X) = f(X) \leq k ,\) that is $-X \in \Acc{k}{f}$.

Item (ix) follows the same reasoning as item (iv).

For item (x), note that the restriction of $D$ to $C$ is both convex and concave, hence the convexity of $\Acc{k}{D} \cap C$ follows the same reasoning that item (v) and the convexity of $(\Acc{k}{D})^\complement \cap C$ from item (vi). For the case when $D$ is additive comonotone, let $X,Y$ be a comonotone pair. Due to \Cref{lemma convex cone comono}, the set $C_{X,Y}$
is a convex cone whose members are all comonotone to one another, and $D$ is additive on $C_{X,Y}$. By the preceding reasoning, the sets $\Acc{k}{D}\cap C_{X,Y}$ and $(\Acc{k}{D})^\complement \cap C_{X,Y}$ are both convex. In particular, if $Z$ is any convex combination of $X$ and $Y$, then $Z\in \Acc{k}{D}\cap C_{X,Y}\subseteq \Acc{k}{D}$ whenever $X,Y \in \Acc{k}{D}$, and similarly $Z\in (\Acc{k}{D})^\complement$ whenever $X,Y \in (\Acc{k}{D})^\complement$.

Item (xi) is just the definition of lower-semicontinuity.
% Item (xii) is a well known fact. To see it suppose by contradiction that $\Acc{k}{f}$ is open. Then, there exist some sequence $X_n \in \Acc{k}{f}$ such $X_n \rightarrow X$ and $X \notin A$, which yields that $f(X_n) \leq k $ for all $n \in \mathbb{N}$ and $f(X) >k$, i.e.\ 
% \[
% f(X) > k \geq \liminf f(X_n) ,
% \]
% an absurd, as $f$ is lower-semicontinuous.

For item (xiii we shall show only for the case $A \coloneqq \Acc{1}{D}$. It holds for general $\Acc{k}{D}$ due to item (i). By continuity of $D$, we have that $A$ is closed whereas the set $B \coloneqq \{X \in \X \colon\, D(X) < 1 \} $ is open. Evidently, $A^\complement $ is open and $B^\complement$ is closed, and the inclusions $B\subseteq\interior A$ and $A^\complement \subseteq \interior (B^\complement)$ hold; {in particular this gives $0\in\interior A$ as $D$ is positive homogeneous, so $A$ is absorbing and $D(X) = \D_{A}(X)<\infty$ for all $X$}. Therefore, $B^\complement \cap A = \{X \in \X\colon\, D(X) = 1 \} = \bd(A)$, where the second equality is yielded by \Cref{ph implies boundary}.
We must show that, for each $X$, the ray $R_X\coloneqq \{\lambda X\colon\,\lambda\in\R_+^* \}$ intersects $\bd (A)$ at most once. For all $X$ such that $D(X) \leq 0$ it is clear that $R_X \subseteq B$ (so $R_X\cap\bd (A) = \varnothing$). It remains to consider the case $0 < D(X) < \infty$. Clearly, $D(\lambda X) = 1$ for $\lambda^{-1} \coloneqq D(X)$, so $R_X\cap\bd(A)$ is nonempty. Moreover, if $\gamma > \lambda$ then clearly $D(\gamma X) > 1$ by positive homogeneity, and if $0< \gamma < \lambda$ then $\gamma X \in B$; in any case $\gamma X \notin \bd A$.
%Moreover, for any $\epsilon > 0$ we have that $f(X (f(X) + \epsilon)^{-1} < 1$. Let $\gamma_X = (f(X))^{-1} $, we clearly have that for $m \in [0, \gamma_X), \gamma X \in B \subseteq \interior (A)$, and for $m \in (\gamma_X, \infty), \gamma X \in A^\complement$. Therefore, it holds that the intersection of $R_X$ with $\bd (A)$ is a singleton, namely $R_X \cap \bd(A) = \{\gamma_X X\} $.

Lastly, for item (xiii) again we shall show only for the case $A \coloneqq \Acc{1}{D}$ and $B \coloneqq \Acc{1}{-D}$, as it holds for general $\Acc{k}{D}$ and $ \Acc{k}{-D}$ due to item (i). Let $Y \in A$ and $X \preceq Y$. Now, remember that for any $Z\in \X , Z \in A $ if and only if $D(Z) \leq 1$. Then we have, by monotoniticy $D$, that $D(X) \leq D(Y) \leq 1$, hence $X \in A$, establishing the anti-monotonicity of $A$. By the same token, let $X \in B$ and $X \preceq Y$. Again, we have that for any $Z\in \X , Z \in B $ if and only if $-D(Z) \leq 1$, and by anti-monotonicity of $-D$ it follows that $1 \geq -D(X) \geq -D(Y)$.
This completes the proof.
\end{proof}

Now, we analyze how some operations on a deviation measure are reflected on its corresponding acceptance set. For a comprehensive theory on combinations of monetary risk measures, see \cite{Righi2018}.

\begin{propo}
Let $ D,D'\colon\X \rightarrow \ra $ be positive homogeneous functionals and $k, \lambda \in \R^*_+$. Then:
\begin{enumerate}[label = (\roman*)]
\item $\Acc{k}{\min(D,D')} = \Acc{k}{D} \cup \Acc {k}{D'}$ and $\Acc{k}{\max(D,D')} = \Acc{k}{D} \cap \Acc {k}{D'}$.
%\item \[ \Acc{k}{f} = \bigcup_{k = c+d, f = g+h} \Acc{c}{g} \cap \Acc{d}{h} \].
\item $X\in \Acc{k}{D}$ if and only if there are non-negative constants $c$ and $d$, and positive homogeneous functions $g$ and $h$ such that $k=c+d$, $D=g+h$ and $X\in \Acc{c}{g}\cap \Acc{d}{h}$. In particular, one has $\Acc{k + \lambda}{D + D'} \supseteq \Acc{k}{D} \cap \Acc{\lambda}{D'}$.
\item $ \Acc{k}{\lambda D} = \lambda^{-1} \Acc{k}{D}$.
\end{enumerate}
\end{propo}

\begin{proof}
For the first item, if $ X \in \Acc{k}{\min(D,D')}$, then $D(X) \leq k$ or $D'(X) \leq k$. That is, $X \in \Acc{k}{D} \cup \Acc {k}{D'}$ Reciprocally, if $X \in \Acc{k}{D} \cup \Acc {k}{D'}$, then we must have $D(X) \leq k$ or $D'(X) \leq k$, so $ \min (D(X) , D'(X) ) \leq k$, which is the same as $X\in \Acc{k}{\min(D,D')}$. The equality $\Acc{k}{\max(D,D')} = \Acc{k}{D} \cap \Acc {k}{D'}$ follows from a similar argument.

Item (ii) is established as follows: assume $X \in \Acc{c}{g} \cap \Acc{d}{h}$, where $ k = c+d$ and $ D = g+h$. Then, by definition, it holds that $g(X) \leq c$ and $h(X) \leq d $. Hence, $D(X) \equiv g(X) + h(X) \leq c+d=k$, which is the same as $X\in \Acc{k}{D}$. For the reverse inclusion, assume $X\in \Acc{k}{D}$. Then, trivially, there are non-negative constants $c\coloneqq k$ and $d\coloneqq 0$, and positive homogeneous functions $g\coloneqq D$ and $h\coloneqq 0$ such that $X\in \Acc{c}{g} \cap \Acc{d}{h} \equiv \Acc{k}{D}$. The last equivalence follows from the fact that $\Acc{d}{h} = \{X\in\X\colon\, 0(X)\leq 0\} \equiv\X$.

Finally, for the last item we have $X \in \Acc{k}{\lambda D}$ if and only if $ D(X) \leq k/\lambda$ if and only if $X\in \Acc{k/\lambda}{D}$. The latter set is equal to $\lambda^{-1} \Acc{k}{D}$ by cref{deviation a.set-item1} in \Cref{deviation a.set}.
\end{proof}

\subsection{Deviation measures: some examples}

A very intelligent Professor once told one of the authors that ``we all think through examples''. Taking the assertion as advice, in this section we discuss some examples of well-known deviation measures and their respective acceptance. The reader will certainly appreciate them.

\begin{ex} Variance ($\sigma^2$): One of the most widely used measures to quantify dispersion. It is defined, for $X\in \X\subseteq L^1$ (recall that we allow for deviation measures to assume $+\infty$), as
\[
\sigma^2 (X) = \E[ (X - \E X )^2 ],
\]
and the associated acceptance sets are given by
\[
\Acc{k}{\sigma^2} = \left\{X \in \X\colon\, \sigma^2 (X) \leq k \right\},\quad k>0.
\]
As the variance is not positive homogeneous, it does not coincide with the \pl\ of $\Acc{1}{\sigma^2}$: indeed, we have
\[
\D_{\Acc{k}{\sigma^2} } (X) =\frac{\sigma (X)}{\sqrt{k}} . 
\]
Also, notice that $\sigma^2(X)<\infty$ if and only if $X\in L^2$.
\end{ex}

\begin{ex}\label{sd deviation} Standard deviation ($\sigma$): The measure used to quantify risk in the seminal paper of \citet{markowitz52}. It has served as inspiration for the class of generalized deviation measures. It is defined, for $X\in \X\subseteq L^1$, as
\[
\sigma (X) = \sqrt{\sigma^2 (X)} = \Vert X - \E X \Vert_2,
\]
and the associated acceptance sets are given by
\[
\Acc{k}{\sigma} = \left\{X \in \X\colon\, \sigma (X) \leq k \right\},\quad k>0.
\]
(Note that $ \Acc{k}{\sigma} =\Acc{k^2}{\sigma^2} $). Interestingly, if $X \preceq_{\mathfrak{D}} Y$ then $ \Vert X - \E X \Vert_2 \geq \Vert Y- \E Y \Vert_2 $; for a detailed proof and more details see \citet{shaked82}. Furthermore, writing $A \coloneqq \Acc{k}{\Vert\cdot\Vert_2}$, we have that
\[
\sigma (X) = k\, \D_{\Acc{k}{\sigma}}(X) = k\, \D_{A+\R} (X), 
\]
where the first equality above follows from \Cref{deviation a.set}, \cref{deviation a.set-item1}, and the second one comes from \Cref{set oper-item4}, together with the identity $k\, \D_{\Acc{k}{2}}= \Vert \cdot\Vert_2$ yielded by \cref{deviation a.set-item1} in \Cref{deviation a.set} and the well-known fact that $\inf_{z\in\R}\Vert X - z\Vert_2 = \Vert X - \E X\Vert_2 = \sigma(X)$ (indeed, $\Vert\cdot\Vert_2$ is the measure of error associated with the standard deviation\footnote{{Importantly, here $\Vert\cdot\Vert_2$ \emph{does not} represent the Euclidian norm.}}). Notice that $\sigma(X)$ is finite if and only if $X\in L^2$. In \Cref{fig variance} bellow, we can see the acceptance set $\Acc{1}{\sigma}$ in blue, (note that $\Acc{1}{\sigma} = \Acc{1}{\sigma^2}$) and the closed unit ball (on the norm $\Vert\cdot\Vert_2 $) in red. The figure also illustrates the relation $\Acc{1}{\Vert\cdot\Vert_2} + \R = \Acc{1}{\sigma}$.

\begin{figure}[H]
\caption{The sub-level sets $\Acc{1}{\sigma}$ (in blue) and $\Acc{1}{\Vert\cdot\Vert_2}$ (in red) in the binary market $\Omega =\{0,1\}$ with $\Pro\{0\} = \nicefrac14$ and $ \Pro\{1\} = \nicefrac34 $.} \label{fig variance}
\begin{center}
\begin{tikzpicture}[scale=1]

\clip (0,0) (-4,-4)rectangle (4,4);

\filldraw[Blue, thick, fill=Blue!20, opacity=.5, rotate=45] (-6,-1.63342)rectangle (6,1.63342);
\filldraw [BrickRed, thick, fill=BrickRed!20, opacity=.5 ](0,0) ellipse (2 cm and 1.1547 cm) ;

\fill (0,0) circle [radius=.05cm] node[anchor=north west,scale=.8]{$0$};

\fill (2,2) circle [radius=.0cm] node[anchor=north west,scale=.8]{$\Acc{1}{\sigma}$};
\fill (0.3,1) circle [radius=.0cm] node[anchor=north west,scale=.8]{$\Acc{1}{\Vert\cdot\Vert_2}$};

\draw[thin,<->] (-4,0) -- (4,0);
\draw[thin,<->] (0,-4) -- (0, 4);

\end{tikzpicture}
\end{center}
\end{figure}
\end{ex}

\begin{ex}
Standard lower-semi-deviation ($\sigma_-$): It is a generalized deviation measure that considers only the negative part of the deviation $X - \E X$. This one is defined, for $X\in\X\subseteq L^1$, as
\[
\sigma_- (X) = \Vert (X - \E X )^- \Vert_2 . 
\]
The corresponding acceptance sets are given by
\begin{align*}
\Acc{k}{\sigma_-} &= \{X\in\X\colon \Vert (X - \E X )^- \Vert_2 \leq k\}
\\&= \left\{X \in \X\colon\, \sigma^2 (X\,|\, \E X \geq X ) \leq k^2/{\mathbb{P} (\E X \geq X )} \right\},\:k>0,
\end{align*}
where $ \sigma^2 (X \,|\, \E X \geq X ) \coloneqq \E\left\{(X - \E\{X \,\big\vert\, \E X \geq X\} )^2 \; |\; \E X \geq X \right\}$ is the conditional variance of $X$ given that $X$ lies in the lower tail of its distribution. Importantly, the set $\Acc{k}{\sigma_{-}}$ contains every random variable whose standard deviation is bounded above by $k$, as $\Vert (X - \E X )^- \Vert_2\leq \Vert X - \E X \Vert_2$ clearly yields $ \Acc{k}{\sigma} \subseteq \Acc{k}{\sigma_-} $. This fact can be seen in \Cref{fig semi variance}, where the acceptance set $\Acc{1}{\sigma}$ of the standard deviation is depicted in blue, and $\Acc{1}{\sigma_-}$ is represented in red. In particular, $\sigma_{-}$ is finite on a subspace which is larger than $\{X\in\X\colon \sigma(X)<\infty\} $.

\begin{figure}[H]
\caption{The sub-level sets $\Acc{1}{\sigma_-}$ (in red) and $\Acc{1}{\sigma}$ (in blue) in the binary market $\Omega =\{0,1\}$ with $\Pro\{0\} = \nicefrac14$ and $ \Pro\{1\} = \nicefrac34 $.} \label{fig semi variance}
\begin{center}
\begin{tikzpicture}[scale=1]

\clip (0,0) (-4,-4)rectangle (4,4);

\filldraw [BrickRed, thick, fill=BrickRed!20, opacity=.5 , rotate=45](-6,-3.265985)rectangle (6,1.885571);

\filldraw[Blue, thick, fill=Blue!30, opacity=.4, rotate=45] (-6,-1.63342)rectangle (6,1.63342);

\fill (0,0) circle [radius=.05cm] node[anchor=north west,scale=.8]{$0$};

\fill (2,2) circle [radius=.0cm] node[anchor=north west,scale=.8]{$\Acc{1}{\sigma}$};
\fill (1,-1.5) circle [radius=.0cm] node[anchor=north west,scale=.8]{$\Acc{1}{\sigma_-}$};

\draw[thin,<->] (-4,0) -- (4,0);
\draw[thin,<->] (0,-4) -- (0, 4);

\end{tikzpicture}
\end{center}
\end{figure}
\end{ex}

\begin{ex}
Lower range deviation ($\LR$): It is the `most conservative' among the class of lower-range dominated generalized deviation measures, defined for $X\in\X\subseteq L^1$ as
\[
\LR (X) = \E[X - \essinf X ], 
\]
with acceptance set
\begin{align*}
\Acc{k}{\LR} &= \left\{X \in \X \colon\, \E X - \essinf X \leq k \right\} \\
&= \left\{X \in \X \colon\, \esssup (-X) \leq \E[-X] + k \right\} . 
\end{align*}
Thus, $\Acc{k}{\LR}$ is comprised of all positions $X$ whose penalized expected loss $\E(-X) + k$ is bounded below by the maximum loss $\esssup(-X)$. Furthermore writing $A = \mathrm{ball}_{\Vert\cdot\Vert_1}(0;\,k)\cap \X_+$, we have that,
\[
\LR (X) = k\, \D_{\Acc{k}{\LR}} (X) = k \D_{A + \R} (X).
\]
The second equality follows from the fact that $k \D_{A}(X)$ assumes $\infty$ for all $X \leq 0$, and equals $\E|X|$ otherwise; thus it coincides with the error function associated with the lower-range deviation --- see \Cref{set oper-item4}.
In \Cref{fig LR}, we can see the acceptance set $\Acc{1}{\LR}$ in blue, and the closed unit ball (on the norm $\Vert\cdot\Vert_1 $) restricted to $\R^2_+$ in red. The fact that $A + \R = \Acc{1}{\LR}$ is clear from this figure.

\begin{figure}[H]
\caption{The sub-level sets $\Acc{1}{\LR}$ (in blue) and $A = \mathrm{ball}_{\Vert\cdot\Vert_1}(0;\,1)\cap \X_+$ (in red) in the binary market $\Omega =\{0,1\}$ with $\Pro\{0\} = \nicefrac14$ and $ \Pro\{1\} = \nicefrac34 $.} \label{fig LR}
\begin{center}
\begin{tikzpicture}[scale=1]

\clip (0,0) (-4,-4)rectangle (4,4);

\filldraw[Blue, thick, fill=Blue!30, opacity=.7, rotate=45] (-6,-2.828427)rectangle (6,0.942809);

\filldraw [BrickRed, thick, fill=BrickRed!40, opacity=.5](0,0)--(0,1.3333333333333)--(4,0)--(0,0);

\fill (0,0) circle [radius=.05cm] node[anchor=north west,scale=.8]{$0$};

\fill (2,2) circle [radius=.0cm] node[anchor=north west,scale=.8]{$\Acc{1}{\LR}$};
\fill (0.2,0.5) circle [radius=.0cm] node[anchor=north west,scale=.8]{$A$};

\draw[thin,<->] (-4,0) -- (4,0);
\draw[thin,<->] (0,-4) -- (0, 4);

\end{tikzpicture}
\end{center}
\end{figure}
\end{ex}

\begin{ex}
Upper range deviation ($\mathrm{UR}$): Defined, for $X\in\X\subseteq L^1$, as
\[
\mathrm{UR} (X) = \esssup X - \E X = \LR(-X),
\]
this measure is the symmetric opposite of LR. Its acceptance set is given by
\[
\Acc{k}{\mathrm{UR}} = \left\{X \in \X\colon\, \esssup X - \E X \leq k \right\} = \left\{X \in \X \colon\, \esssup X \leq \E X + k \right\} . 
\]
Furthermore, writing $A = \mathrm{ball}_{\Vert\cdot\Vert_1}(0;\,k)\cap \X_{-}$ we have that
\[
\mathrm{UR} (X) = k\, \D_{\Acc{k}{\mathrm{UR}}} (X) = k\, \D_{A + \R} (X),
\]
where the second equality follows from the same reasoning as the one for $\LR$.\end{ex}

\begin{ex}
Full range deviation ($\mathrm{FRD}$): Can be considered the most extreme generalized deviation measure, defined for $X\in\X = \{X \in L^0 \colon\, \essinf X < \infty \text{ or } \esssup X > - \infty \} $ as
\[
\mathrm{FRD} (X) = \esssup X - \essinf X, 
\]
with acceptance set
\[
\Acc{k}{\mathrm{FRD}} = \left\{X \in \X \colon\, \esssup X \leq k + \essinf X \right\}. 
\]
Furthermore, writing $A = \Acc{k}{\Vert\cdot\Vert_\infty}$ we have that
\[
\mathrm{FRD}(X) = k\, \D_{\Acc{k}{\mathrm{FRD}}} (X) = 2\,k\, \D_{A+\R}(X),
\]
where the second equality is due to \Cref{set oper-item4} and the fact that $2\,k\,\f (X) = 2 \Vert X \Vert_\infty$, which is the error function associated with the full range deviation.
Note that $\mathrm{FRD} (X) < \infty $ if and only if $X \in \Li$.
In \Cref{fig LR}, we can see the acceptance set $\Acc{1}{\mathrm{FRD}}$ in blue, and the closed unit ball (on the norm $\Vert\cdot\Vert_\infty $), scaled down in half, in red. Clearly, $\Acc{0.5}{\Vert\cdot\Vert_\infty} + \R = \Acc{1}{\LR}$.

\begin{figure}[H]
\caption{The sub-level sets $\Acc{1}{\mathrm{FRD}}$ (in blue) and $ A = \Acc{0.5}{\Vert\cdot\Vert_\infty}$ (in red) in the binary market $\Omega =\{0,1\}$ with $\Pro\{0\} = \nicefrac14$ and $ \Pro\{1\} = \nicefrac34 $.} \label{fig FR}
\begin{center}
\begin{tikzpicture}[scale=1]

\clip (0,0) (-4,-4)rectangle (4,4);

\filldraw[Blue, thick, fill=Blue!30, opacity=.7, rotate=45] (-6,-0.70710678118)rectangle (6,0.70710678118);

\filldraw [BrickRed, thick, fill=BrickRed!40, opacity=.5](-0.5,-0.5)rectangle (0.5,0.5);

\fill (0,0) circle [radius=.05cm] node[anchor=north west,scale=.8]{$0$};

\fill (2,2) circle [radius=.0cm] node[anchor=north west,scale=.8]{$\Acc{1}{\LR}$};
\fill (0.2,0.5) circle [radius=.0cm] node[anchor=north west,scale=.8]{$A$};

\draw[thin,<->] (-4,0) -- (4,0);
\draw[thin,<->] (0,-4) -- (0, 4);

\end{tikzpicture}
\end{center}
\end{figure}
\end{ex}

\begin{ex}
Expected shortfall deviation ($\mathrm{ESD}$): A generalized deviation measure derived from the (standard) expected shortfall. It is defined, for $X\in\X\subseteq L^1$ and $0<\alpha\leq1$, by $\mathrm{ESD}_\alpha (X) = \mathrm{ES}_\alpha (X - \E X) $ with,
\[\mathrm{ES}_\alpha(X) = -\int^\alpha_0 \frac{1}{\alpha} F^{-1}_X (t) \, \dd t,
\]
and $\mathrm{ESD}_\alpha (X) = \E X-\operatorname{ess}\inf X=\LR(X)$ for $\alpha=0$. Note that if we take $\gamma = 1- \alpha$ we have that $\mathrm{ES}_\alpha(X) = \int_\gamma^1 \frac{1}{1-\gamma} F^{-1}_X (t) \,\dd t $. Furthermore, if $F_X$ is continuous, then the following representation also holds. 
\[
\mathrm{ESD}_\alpha(X) = \mathrm{ES}_\alpha (X- \E X) 
\equiv - \E\big(X-\E X \; | \; X \leq \D_X^{-1}(\alpha) \big)
= \E(X) - \E\big(X\,\vert\, X\leq \D_X^{-1}(\alpha)\big)
\]
with acceptance set
\begin{align*}
\Acc{k}{\mathrm{ESD}_\alpha}
&=\{X\in\X\colon\, 
k - \mathrm{ES}_\alpha (X) \geq \E X 
\}
\end{align*}
If we let the \emph{Koenker-Bassett error} be defined as $\mathrm{KB}_\alpha (X) = \E \left[\alpha^{-1}{(1-\alpha)} X^- + X^+ \right]$, which is the error function associated with the $\mathrm{ESD}$, then we have $\mathrm{KB}_\alpha = k\, \D_{A}$, with $A = \Acc{k}{\mathrm{KB}_\alpha}$. Hence --- by \Cref{set oper-item4} --- it holds that
\[
\mathrm{ESD}_\alpha (X) = k\, \D_{\Acc{k}{\mathrm{ESD}_\alpha}} (X) = k\, \D_{A + \R} (X).
\]
\begin{figure}[H]
\caption{The sub-level sets $\Acc{1}{\mathrm{ESD_{\alpha}}}$ (in blue) and $A = \Acc{1}{\mathrm{KB}_{\alpha}}$ (in red), with $\alpha=0.1$, in the binary market $\Omega =\{0,1\}$ with $\Pro\{0\} = \nicefrac14$ and $ \Pro\{1\} = \nicefrac34 $.} \label{fig ESD}
\begin{center}
\begin{tikzpicture}[scale=0.5]

\clip (0,0) (-8,-8)rectangle (8,8);

\filldraw[Blue, thick, fill=Blue!30, opacity=.7, rotate=45] (-15,-2.82)rectangle (15,0.9428);
\filldraw[BrickRed, thick, fill=BrickRed!30, opacity=.7](-0.4444444,0)--(0,-0.148148)--(4,0)--(0,1.3333333)--(-0.4444444,0) ;

\fill (0,0) circle [radius=.05cm] node[anchor=north west,scale=.8]{$0$};

\fill (0.2,0.8) circle [radius=.0cm] node[anchor=north west,scale=.8]{$A$};
\fill (2,2) circle [radius=.0cm] node[anchor=north west,scale=.8]{$\Acc{1}{\mathrm{ESD_{\alpha}}}$};

\draw[thin,<->] (-8,0) -- (8,0);
\draw[thin,<->] (0,-8) -- (0,8);

\end{tikzpicture}
\end{center}
\end{figure}
\end{ex}

\bibliography{minkowski.bib}

\appendix
\section{Auxiliary results}\label{apen A}

We begin with a result which we use many times in throughout the paper. It relates star-shapedness with the fact that the infimum in the definition of the \pl\ is taken over an interval.

\begin{lemma}\label{star}
Let $A\subseteq\X$ and $X\in\X$. Then
\begin{enumerate}[label = (\roman*)]
\item \label{lemma-item3} If $\A$ contains a cone $M$, then $\f (X) = 0$, for all $X \in M$; in particular as $\{0\} $ is a cone, if $0 \in A$ then $\f(0)=0$.
\item $\f(X)=\infty$ if and only if $\{m\in\R_+^*\colon\,m^{-1}X\in A\} = \varnothing$ if and only if $\{m\in\R_+^*\colon\,m^{-1}X\notin A\} = \R_+^*$.
\end{enumerate}
Moreover, if $A$ is star-shaped, then
\begin{enumerate}[label = (\roman*), resume]
\item\label{star item f=0} $\f(X) = 0$ if and only if $\{m\in\R_+^*\colon\,m^{-1}X\in A\} = \R_+^*$ if and only if $\{m\in\R_+^*\colon\,m^{-1}X\notin A\} = \varnothing$.
\end{enumerate}
If in addition $0<\f(X)<\infty$, then one of the following holds:
\begin{enumerate}[label = (\roman*), resume]
\item $\{m\in\R_+^*\colon\,m^{-1}X\in A\} = [\f(X),\infty)$ and $\{m\in\R_+^*\colon\,m^{-1}X\notin A\} = (0,\f(X))$ (this is true in particular when $A$ is closed).
\item $\{m\in\R_+^*\colon\,m^{-1}X\in A\} = (\f(X),\infty)$ and $\{m\in\R_+^*\colon\,m^{-1}X\notin A\} = (0,\f(X)]$ (this is true in particular when $A$ is open).
\end{enumerate}
\end{lemma}

\begin{proof}
The first item is from lemma 5.49 of \cite{aliprantis06}). The second is immediate. For the remaining assertions, let $T_X(m)\coloneqq m^{-1}X$ for $m\in\R_+^*$.  $T_X$ is clearly a continuous mapping from $\R_+^*$ to $\X$. We have $T_X^{-1}(A) = \{m\in\R_+^*\colon\,m^{-1}X\in A\}$ and similarly $T_X^{-1}(A^\complement) = \{m\in\R_+^*\colon\,m^{-1}X\notin A\}$. Assume now that $A$ is star-shaped and $m\in T_X^{-1}(A)$. Then, if $m'>m$, we have $m'\in T_X^{-1}(A)$ as well. This establishes that $T_X^{-1}(A)$ is always an interval with $\infty$ as its right endpoint, and by definition the left endpoint is $\f(X)$, thus establishing (ii), (iii) and (iv), where the topological assertions follow by continuity of $T_X$.
\end{proof}

We then have the following direct corollary on the relation between gauge and co-gauge.

\begin{coro}\label{coro cogauge}
Let $A\subseteq\X$ be star-shaped. Then the equality
\begin{equation} \label{coga igu}
\f (X) = \cog_{A^\complement} (X)
\end{equation} 
holds for all $X\in\X$.
\end{coro}

\begin{lemma} \label{star-shaped no cone is radially bounded}
Let $A\subseteq\X$. If $A$ is closed, star-shaped, and contains a proper cone with vertex at some constant $x \in \R$, then $A$ is not radially bounded. Hence, if $A$ is closed, star-shaped, and radially bounded, then every proper cone with vertex at a constant intersects $A^\complement$.
%then $A$ does not contain any cone with vertex at some $c \in \R$.
\end{lemma}
\begin{proof}
As $A$ contains a proper cone with vertex at some constant $x\in\R$, there exists a non-zero $X \in \X$ such that $\{x + \lambda X\colon\, \lambda \ge0 \} \subseteq A$. As $A$ is star-shaped, we have that $ k(x + \lambda X) \in A$ for all $ k \in [0,1]$ and all $\lambda \ge0$; in particular, taking $\lambda=1/k$, we have $kx + X\in A$ for all $k\in(0,1]$ and, as $A$ is closed, $X = \lim_{k\downarrow0}kx+X\in A$. To conclude that $A$ is not radially bounded, it is sufficient to show that there is no $ \delta_X>0$ such that $\delta X \notin A $ for $\delta \geq \delta_X$. So, let us fix an arbitrary $\delta_X>0$ and put $k_n=1/n$ and let $\lambda_n = \delta_X / k_n$. As $A$ is closed, we have $\lim_{n \to \infty}(k_n x + k_n \lambda_n X) \in A $. Now, clearly the preceding limit equals $\delta_X X$ and so, as $\delta_X$ was chosen arbitrarily, we can conclude that $A$ is not radially bounded.
\end{proof}

\begin{remark}
A quick inspection of the proof of \Cref{star-shaped no cone is radially bounded} tells us that it remains true even when the vertex $x$ is not assumed to be a constant. In any case, we opt to state it for constant vertices since this is the case which is employed in the text.
\end{remark}

The next lemma shows that positive homogeneity is also a sufficient condition ensuring that an arbitrary positive homogeneous functional $f$ (which does not assume negative values) is the \pl\ of \emph{some} subset of $\X$. We opt to state the result as it appears in \cite{aliprantis06}, where it is assumed at the outset that $\mathrm{range}(f)\subseteq \R_+$. This assumption can be easily dropped; if so, the set $V$ appearing in \Cref{lemma2} is no longer (necessarily) absorbing. Instead, in this case the condition $0\in V$ must hold.

\begin{propo} \label{lemma2} ({Lemma} 5.50 and Theorem 5.52 of \cite{aliprantis06} )
Let $ A , \B \subseteq \X$ be non-empty, and let $f\colon\X \rightarrow \R_+$ be an arbitrary %positive
function. Then the following holds
\begin{enumerate}[label = (\roman*)]
\item \label{lemma-item6} $f$ is positive homogeneous if and only if it is the \pl\ of an absorbing set, in which case for every $V \subseteq \X$ satisfying
\[
\{X \in \X\colon\, f (X) < 1 \} \subseteq V \subseteq \Acc{1}{f},
\]
we have $\D_V=f$.
\item\label{lemma-item7} $f$ is sub-linear (positive homogeneous and convex) if and only if it is the \pl\ of a convex absorbing set $V$, in which case we may take $V = \Acc{1}{f}$.
\item $f$ is sub-linear and symmetric if and only if it is the \pl\ of a symmetric, convex, absorbing set $V$, in which case we may take $V = \Acc{1}{f}$.
\item\label{continuidades-item1} $f$ is sub-linear and lower-semicontinuous if and only if it is the \pl\ of an absorbing, closed convex set $V$, in which case we may take $V = \Acc{1}{f}$.
\item\label{continuidades-item2} $f$ is sub-linear and continuous if and only if it is the \pl\ of a convex neighborhood $V$ of zero, in which case we may take $V = \Acc{1}{f}$.
\item \label{lemma-item9} $f$ is sub-linear, symmetric and continuous if and only if it is the \pl\ of a unique closed, symmetric and convex neighborhood $V$ of zero, namely
$V = \Acc{1}{f}$.
\end{enumerate}
\end{propo}

\begin{remark}
A locally convex topology is a topology generated by a family of seminorms. In particular, the neighborhood base at zero is given by the collection of all $\Acc{k}{p}$, with $k >0$ and $p$ {belonging to some collection of seminorms}. {Now, \Cref{lemma2} \cref{lemma-item9} actually tells us that each $p$ is the \pl\ of some unique closed, symmetric, convex neighborhood $A$ of zero, namely $A = \Acc{1}{p}$, with $p = \f$. Distinctively, Theorem 5.73 of \cite{aliprantis06} tell us that any locally convex topology is generated by the family of gauges of the symmetric convex closed neighborhoods of zero.}
\end{remark}

As convexity plays a central role in risk analysis and optimization, it is a relief to see that taking the convex hull of an acceptance set translates as expected into the corresponding \pl.

\begin{propo}\label{conv envelop}
Let $ A\subseteq\X$. If $0\in A$, then the \pl\ of the closed convex hull of $A$ is equal to the convex envelope of the \pl\ of $A$, i.e.,\ one has
\[
\D_{\clconv A} (X) = \conv\f (X)
\]
for all $X\in\X$.
\end{propo}

\begin{proof}
 First, notice that any lower-semicontinuous sub-linear function $g \geq 0$ that is dominated by $\f$ can be written as $g = \D_{C}$, with $C$ a closed convex set given by $C = \Acc{1}{g}\supseteq \Acc{1}{\f} \supseteq A$ (see \Cref{lemma}, \Cref{lemma2} and \Cref{deviation a.set}, {where the absorbing condition can be dropped by letting $g$ assume $+\infty$}). Reciprocally, if $C$ is any closed convex set such that $A\subseteq C$, then the sub-linear function $g\coloneqq \D_C\geq 0$ is dominated by $\f$. In summary, there is a one-to-one correspondence between the class $\mathfrak{S}_+(\f)$ comprised of all lower-semicontinuous sub-linear mappings $g\colon\X\to\R_+\cup\{+\infty\}$ dominated by $\f$ and the class $\mathfrak{C}$ comprised of all closed convex sets $C\supseteq A$. Therefore, since by definition $\clconv A = \bigcap_{C\in\mathfrak{C}} C$, an easy generalization of \cref{lemma-item5} in \Cref{lemma} entails
 \[
 \D_{\clconv A} (X) = \sup_{C\in\mathfrak{C}} \D_{C} (X) = \sup_{g\in \mathfrak{S}_+(\f)}g(X).
 \]
 
 Now, let $\mathfrak{S}(f)$ be the set of all lower-semicontinuous sub-linear functions dominated by a mapping $f$, and $\mathfrak{A}(f)$ the set of all continuous affine functions dominated by $f$.
 The supremum over $\mathfrak{S}_+(\f)$ in the above expression corresponds to the supremum over all lower-semicontinuous sub-linear functions with values in $\R_+\cup\{+\infty\}$ that are dominated by $\f$ and it clearly coincides with the supremum over all (not necessarily positive) lower-semicontinuous sub-linear functions that are dominated by $\f$. That is, we have
 \[
 \sup_{g\in \mathfrak{S}_+(\f)}g(X) = \sup_{g \in \mathfrak{S}(\f)} g(X).
 \]
 As any lower-semicontinuous sub-linear function can be written as the supremum of the continuous affine functions that it dominates (by taking its convex envelope), we have that
 \begin{align*}
  \sup_{f \in \mathfrak{S}(\f)} \sup_{g\in\mathfrak{A}(f)} g(X)
  &= \sup \Big\{g(X)\colon\, g\in\bigcup\nolimits_{f \in \mathfrak{S}(\f)} \mathfrak{A}(f) \Big\} 
  \\ & = \sup \Big\{g(X)\colon\, g\in\mathfrak{A}(\f)\Big\} .
  \\ &= \conv\f (X)
 \end{align*}
 and this completes the proof.
\end{proof}

\begin{remark}
 If the convex envelope of a function $f$ is defined as the supremum over the (not necessarily continuous) affine functions that it dominates, then $\conv f $ is not necessarily lower-semicontinuous. Nevertheless, the proposition above can easily be adapted to yield the equality $ \conv \D = \D_{\conv A}$ by changing convex, closed sets for convex sets and dropping all the requirements of continuity over $g,f$ and the affine functions appearing in the proof.
\end{remark}

\begin{lemma}\label{ph implies boundary}
Let $f\colon\X\to \R\cup\{\infty\}$. If $f$ is positive homogeneous, then the set $E \coloneqq \{X\in\X\colon\,f(X)=1\}$ has empty interior.
\end{lemma}
\begin{proof}
Let us proceed by contraposition by showing that if $E$ has non-empty interior, then $f$ is not positive homogeneous. Assume, then, that $X\in\interior E$, and let $V$ denote an open neighborhood of $X$ with $V\subseteq E$. By continuity of scalar multiplication, for small enough $u>0$ we have $(1+u)X\in V\subseteq E$. But then $f((1+u)X) = 1 < (1+u) f(X)$, so $f$ is not positive homogeneous.
\end{proof}

The following result characterizes polar sets through the \pl. {Recall that, by definition, the polar of a set $A\subseteq\X$ is given by
\(
A^{\odot} = \{X' \in \X'\colon\, \langle X, X' \rangle \leq 1\text{ for all }X \in A \}.
\)}

\begin{propo} \label{polar alter}
Let $A$ be star-shaped. Then it holds that
\begin{equation}\label{polar-alternative-representation}
A^\odot = \{X' \in \X' \colon\, \langle X, X' \rangle \leq \f (X) \text{ for all }X \in \X \} .
\end{equation}
\end{propo}
\begin{proof}
Notice that we can write
\[A^\odot = \{X'\in\X'\colon\, \langle X,X'\rangle\leq 1 \text{ for all $X\in A$} \} = B_0\cap B\cap B_\infty,\]
where
\begin{align*}
&B_0 = \{X'\in\X'\colon\, \langle X,X'\rangle\leq 1 \text{ for all $X\in A$ such that $ \f (X) = 0$}\} ,\\
&B =\{X'\in\X'\colon\, \langle X,X'\rangle\leq 1 \text{ for all $X\in A$ such that $0< \f (X)<\infty$}\} ,\\
&B_\infty =\{X'\in\X'\colon\, \langle X,X'\rangle\leq 1 \text{ for all $X\in A$ such that $ \f (X)=\infty$}\} .
\end{align*}
Similarly, we can write the right-hand side in \eqref{polar-alternative-representation} as
\[\{X' \in \X' \colon\, \langle X, X' \rangle \leq \f (X) \text{ for all }X \in \X \} = B_0^*\cap B^*\cap B_\infty^*,\]
where
\begin{align*}
&B_0^* = \{X'\in\X'\colon\, \langle X,X'\rangle\leq 0 \text{ for all $X\in \X$ such that $\f (X) = 0$}\} ,\\
&B^* =\{X'\in\X'\colon\, \langle X,X'\rangle\leq \D_A(X) \text{ for all $X\in \X$ such that $0< \f (X)<\infty$}\} ,\\
&B_\infty^* =\{X'\in\X'\colon\, \langle X,X'\rangle\leq \infty \text{ for all $X\in \X$ such that $ \f (X)=\infty$}\} .
\end{align*}
Clearly $B^{\vphantom{*}}_\infty = B_\infty^* = \X'$ since $B_\infty$ is defined by a vacuous sentence and the upper bound $\f(X)=\infty$ in $B_\infty^*$ is non-binding. Thus, to establish the proposition it suffices to show that $B_0 = B_0^*$ and $B = B^*$.

For the equality $B_0 = B_0^*$, suppose $X'\in B_0$ and let $X\in \X$ be such that $\f(X) = 0$. If $\langle X, X'\rangle\leq 0$ then there is nothing to show as in this case $X'\in B_0^*$. If $\langle X,X'\rangle \geq 0$, then --- as $\f$ is positive homogeneous --- we have $\f(\lambda X) = 0$ for all $\lambda>0$ and, by assumption, $\langle \lambda X,X'\rangle\leq 1$ for all $\lambda>0$, which necessarily entails $\langle X,X'\rangle = 0$. Thus, $B_0\subseteq B_0^*$. That $B_0^*\subseteq B_0$ is obvious. Hence, $B_0 = B_0^*$

For the equality $B = B^*$, suppose $X'\in B$ and let $X\in \X$ be such that $0<\f(X)<\infty$. Writing $Y = X/\f(X)$, we have $Y\in A$ by \cref{propo2-item3} in \Cref{propo2} and $0<\f (Y)<\infty$ by positive homogeneity. Thus $\langle Y,X'\rangle\leq 1$ or, which is the same, $\langle X, X'\rangle \leq \D_A(X)$. The preceding argument shows that, $B\subseteq B^*$. Reciprocally, suppose $X'\in B^*$ and let $X\in A$ be such that $0<\f (X)<\infty$. Writing $Y = \D_A(X)X\in\X$, then again positive homogeneity entails $0<\f (Y)=\f (X)^2<\infty$. Thus, $\langle Y,X'\rangle\leq \D_A(Y)$ or, equivalently, $\langle \D_A(X)X,X'\rangle\leq \D_A(X)^2$, from which we deduce that $\langle X,X'\rangle\leq 1$ since $\f(X)\leq 1$. Therefore, $B^*\subseteq B$, which establishes the equality in $B = B^*$.
\end{proof}

\begin{ex}\label{bipolar-counterexample}
Let $\X = L^1$, so that $\X' = L^\infty$ and $\X''= \mathrm{ba}$, where $\mathrm{ba}$ is the set of all finitely additive measures on $(\Omega,\mathfrak{F})$ that are absolutely continuous w.r.t.\ $\Pro$. With the dual pair $\langle L^1,L^\infty\rangle$ in mind, if $A$ is the unit ball in $\X$, then clearly $A^\odot \supseteq \mathrm{ball}(L^\infty)$. To see that the converse inclusion $A^\odot \subseteq \mathrm{ball}(L^\infty)$ also holds, notice that if $X'\in\X'$ is such that $\Vert X'\Vert_\infty > 1$ then, since the random variable $X=\mathbb{I}_{[X'> \lambda]}/\Pro[X'>\lambda]$ belongs to $\mathrm{ball}(L^1)$ for any conformable $1<\lambda<\Vert X'\Vert_\infty$, we have for such an $X$
$$
\langle X, X'\rangle = \frac{1}{\Pro[X'>\lambda]}\int_{[X'>\lambda]} X'\,\mathrm{d}\Pro\geq \frac{1}{\Pro[X'>\lambda]}\int_{[X'>\lambda]}\lambda\,\mathrm{d}\Pro > 1,
$$
hence $X'\notin A^\odot$. Fix $B \coloneqq A^\odot$ and, now with the dual pair $\langle L^\infty, \mathrm{ba}\rangle$ in mind, notice that given any $X''\in\mathrm{ba}$ with total variation less than 1, clearly one has $\langle X'', X'\rangle\leq 1$ for all $X'\in B$. That is, $X''\in B^\odot$. However, since $L^1$ is not reflexive, not every such $X''$ is the image of an $X\in L^1$ via the canonical embedding. Therefore, $(A^\odot)^\odot \supsetneq A^{\odot\odot}$.
\end{ex}

\begin{ex}\label{law couter exemple}
Let $\Omega = \{0,1\}^\N$ be the Bernoulli space comprised of all sequences of $0$'s and $1$'s, that is, the generic element $\omega\in\Omega$ is of the form $\omega = (\omega_1,\omega_2, \dots)$ with $\omega_n\in\{0,1\} $ for all $n$. The probability measure $\Pro$ is defined, for each $n$ and each $n$tuple $x_1,\dots,x_n\in\{0,1\}$, via
\[
\Pro\{\omega\in\Omega\colon\,\omega_1=x_1,\dots,\omega_n=x_n,\omega_{n+1}\in\R,\omega_{n+1}\in\R,\dots\} = 1/2^n
\]
Now define $X_n(\omega) = n\times\mathbb I(\omega_n=1)$, and put $B = \{X_1,X_2,\dots\} $. Such $B$ is radially bounded, since for any fixed element $X_n\in B$, there is only one element of $B$ in the direction $\overline{0\,X_n}$. However, $\mathcal{L_B}$ is not radially bounded: indeed, since $nX_1 =_d X_n$, we have that $nX_1\in\mathcal{L}_B$ for all $n$, and thus $\cal L_B$ is not radially bounded in the direction of $X_1$. Similarly, $X_2/2 =_d X_{2n}/2n $ and thus we have $nX_2\in\cal L_B$ for all $n$, and so on.
\end{ex}

\section{Figures}\label{apen figs}

\begin{figure}[H]
\caption{Representation of the \pl\ $\f$ of a set $A$.}\label{fig:minkowski-gauge}
\begin{center}
\begin{tikzpicture}[scale=2]

\clip (0,0) (-2,-2)rectangle (2,2);
\filldraw[thick, fill=black!20, fill opacity=.5, rotate = 45] circle (1) ;

\fill (1,1.5) circle [radius=.03cm] node[anchor=north west,scale=.8]{$X$};
\fill (0.555,0.828) circle [radius=.03cm] node[anchor=north west,scale=.8]{$\dfrac{X}{\f(X)}$};
\draw[dashed,->, Red](0,0)--(1.2,1.8);
\fill (0,0) circle [radius=.03cm] node[anchor=north west,scale=.8]{$0$};
\draw (-.4,0.4) node[]{$A$};
\draw[thin, ->] (-1.5,0) -- (1.5,0);
\draw[thin, ->] (0,-1.5) -- (0, 1.5);
\end{tikzpicture}
\end{center}
\end{figure}

\begin{figure}[H]
\caption{A set $A$ which is radially bounded and strongly star-shaped. The ray $R_X$ is represented by the dashed line in red, which clearly ``leaves'' the set (as any such ray).}\label{fig radially bounded}
\begin{center}
\begin{tikzpicture}[scale=2]

\clip (0,0) (-2,-2)rectangle (2,2);

\fill[ fill=black!20, fill opacity=.5, rotate = 45] (-3,0) rectangle (3,1);

\fill (0,0) circle [radius=.03cm] node[anchor=north west,scale=.8]{$0$};
\draw[dashed,<->](-2,-2)--(2,2);
\draw[thin,->] (-2,0) -- (2,0);
\draw[thin,->] (0,-2) -- (0, 2);
\draw[<->, rotate = 45](-3,1)--(3,1) ;

\fill (-0.6,-0.1) circle [radius=.02cm] node[anchor=north west,scale=.8]{$X$};
\fill (-1.692,-0.282) circle [radius=.02cm] node[anchor=north west,scale=.8]{$\delta_X X$};
\fill (-1.92,-.32) circle [radius=.02cm] node[anchor=north west,scale=.8]{};
\fill (-1.95,-.12) circle [radius=.0cm] node[anchor=north west,scale=.8]{$\delta X$};
\draw[dashed,->, Red, rotate = 0](0,0)--(-1.992,-0.332);
\fill (0,0) circle [radius=.02cm] node[anchor=north west,scale=.8]{$0\in A$};
\fill (1,1) circle [radius=.02cm] node[anchor=north west,scale=.8]{$1\notin A$};

\draw[thin,<->] (-2,0) -- (2,0);
\draw[thin,<->] (0,-2) -- (0, 2);
\end{tikzpicture}
\end{center}
\end{figure}
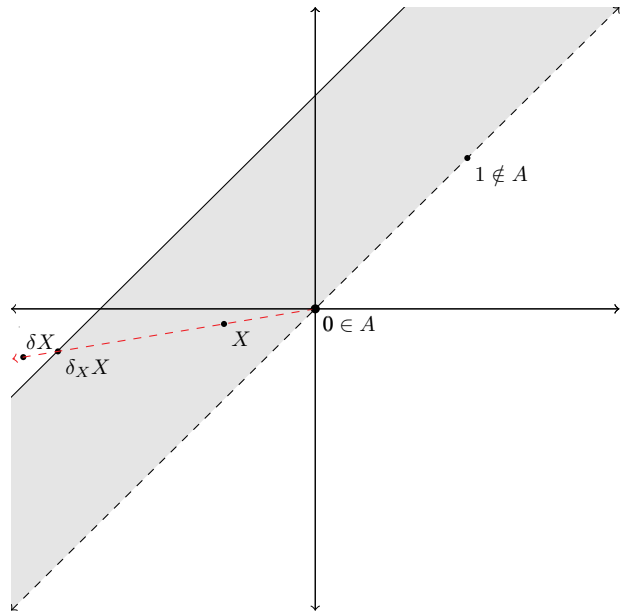

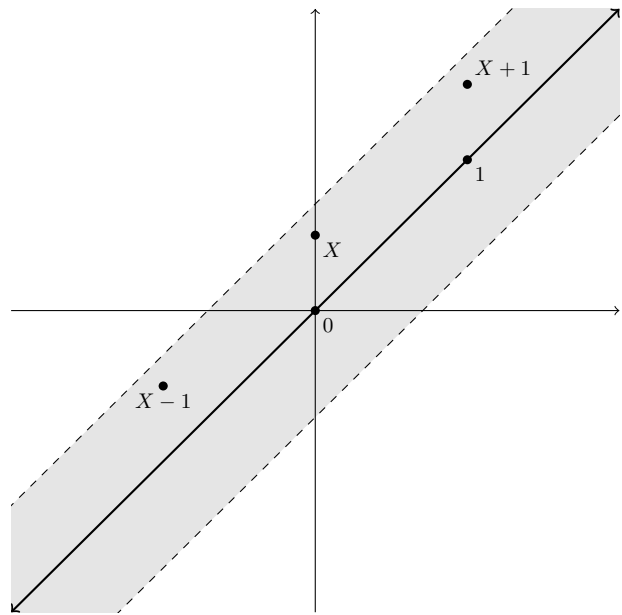
\begin{figure}[H]
\caption{A set $A$ which is absorbing, radially bounded at non-constants, stable under scalar addition and star-shaped. The subspace $\R$ of constant random variables is represented by the thick black diagonal.}\label{fig stable scalar addition}
\begin{center}
\begin{tikzpicture}[scale=2]

\clip (0,0) (-2,-2)rectangle (2,2);
\fill[ fill=black!20, fill opacity=.5, rotate = 45] (-3,-0.5) rectangle (3,0.5);
\draw[<->, rotate = 45, dashed](-3,-.5)--(3,-.5) ;
\draw[<->, rotate = 45, dashed](-3,.5)--(3,.5) ;

\fill (1,1.5) circle [radius=.03cm] node[anchor=south west,scale=.8]{$X+1$};
\fill (0,0.5) circle [radius=.03cm] node[anchor=north west,scale=.8]{$X$};
\fill (-1,-0.5) circle [radius=.03cm] node[anchor=north,scale=.8]{$X-1$};
%\draw[dashed, red](-1,-0.5)--(1,1.5);
\fill (0,0) circle [radius=.03cm] node[anchor=north west,scale=.8]{$0$};
\fill (1,1) circle [radius=.03cm] node[anchor=north west,scale=.8]{$1$};
\draw[thick,<->](-2,-2)--(2,2);
\draw[thin,->] (-2,0) -- (2,0);
\draw[thin,->] (0,-2) -- (0, 2);
\end{tikzpicture}
\end{center}
\end{figure}

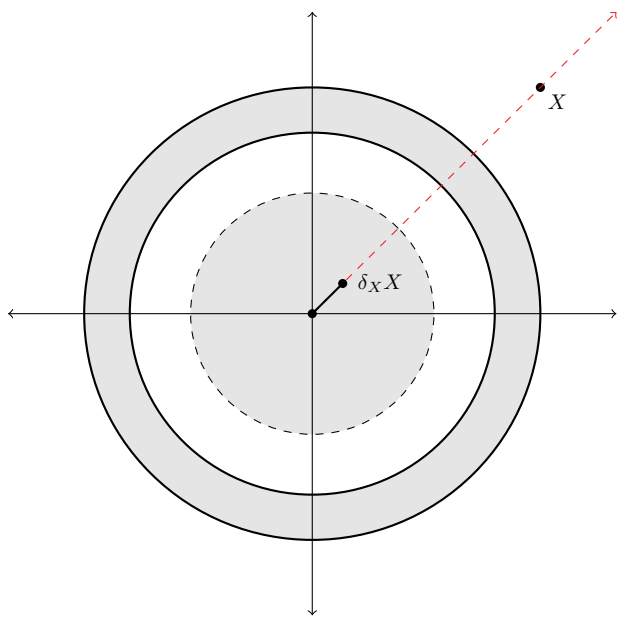
\begin{figure}[H]
\caption{A set $A$ which is absorbing and radially bounded. Notice that $\delta_X$ is not uniquely defined.} \label{fig absorbing set}
\begin{center}
\begin{tikzpicture}[scale=2]

\clip (0,0) (-2,-2)rectangle (2,2);

\filldraw[ fill=black!20, fill opacity=.5, rotate = 45, thick] circle (1.5);

\filldraw[ fill=white, thick] circle (1.2);

\filldraw[ fill=black!20, fill opacity=.5, rotate = 45, dashed] circle (0.8);

\fill (1.5,1.5) circle [radius=.03cm] node[anchor=north west,scale=.8]{$ X$};

\fill (0,0) circle [radius=.03cm] ;
%\fill (0.02,0.02) circle [radius=.03cm] node[below=.1cm,scale=.8]{$ t X$};
\fill (0.2,0.2) circle [radius=.03cm] node[right=.1cm,scale=.8]{$ \delta_X X$};
\draw[dashed,->, Red](0,0)--(2,2);

\draw[thin,<->] (-2,0) -- (2,0);
\draw[thin,<->] (0,-2) -- (0, 2);
\draw[thick,-] (0,0) -- (0.2, 0.2);
\end{tikzpicture}
\end{center}
\end{figure}

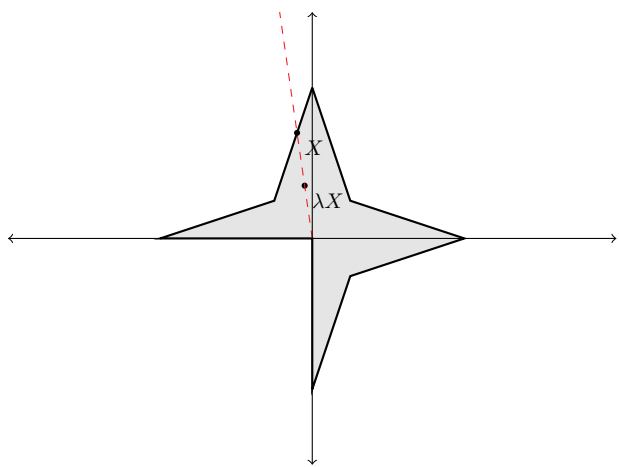
\begin{figure}[H]
\caption{A set $A$ which is star-shaped set and radially bounded.} \label{fig star-shaped}
\begin{center}
\begin{tikzpicture}[scale=2]

\clip (0,0) (-2,-1.5)rectangle (2,1.5);

\draw[ fill=black!20, fill opacity=.5, thick] (0,1)--(-0.25,0.25)--(-1,0)--(0,0)-- (0,-1)--(0.25,-0.25)--(1,0)--(0.25,0.25)--(0,1);

\fill (-0.1,0.7) circle [radius=.02cm] node[anchor=north west,scale=.8]{$ X$};
\fill (-0.05,0.35) circle [radius=.02cm] node[anchor=north west,scale=.8]{$ \lambda X$};
\draw[dashed,->, Red](0,0)--(-0.28571428571,2);

\draw[thin,<->] (-2,0) -- (2,0);
\draw[thin,<->] (0,-1.5) -- (0, 1.5);
\end{tikzpicture}
\end{center}
\end{figure}

\begin{figure}[H]
\caption{A set $A$ which is strongly star-shaped, with $0\in\bd(A)$.\label{fig:strong-star}}
\begin{center}
\begin{tikzpicture}[scale=2]
\draw[thin,->] (-1.5,0) -- (2,0);
\draw[thin,->] (0,-.5) -- (0, 2.5);
\filldraw[thick, fill=black!20, fill opacity=.5, draw=black] (0,1) circle [radius=1cm];
\draw (-.5,1.2) node[]{$A$};
\draw[color=Red, dashed, ->] (0,0) -- (3,2);
\draw[thick,->] (0,0) -- (1.5,1) node[anchor=north west]{$X$};
\fill (0,0) circle [radius=.03cm] node[anchor=north west,scale=.8]{$0\in \mathrm{bd}(A)$};
\end{tikzpicture}
\end{center}
\end{figure}
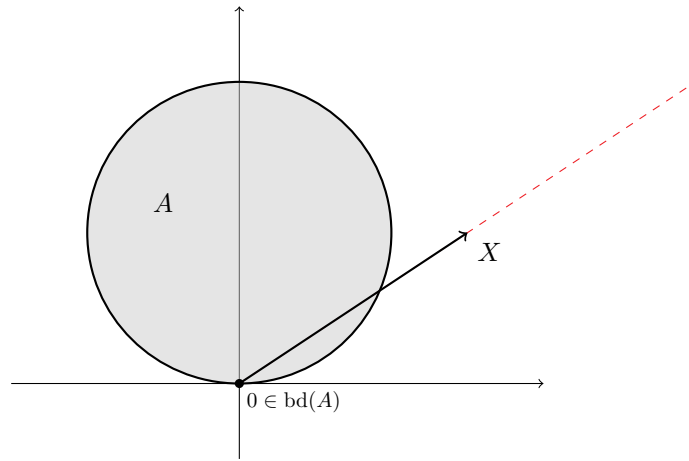

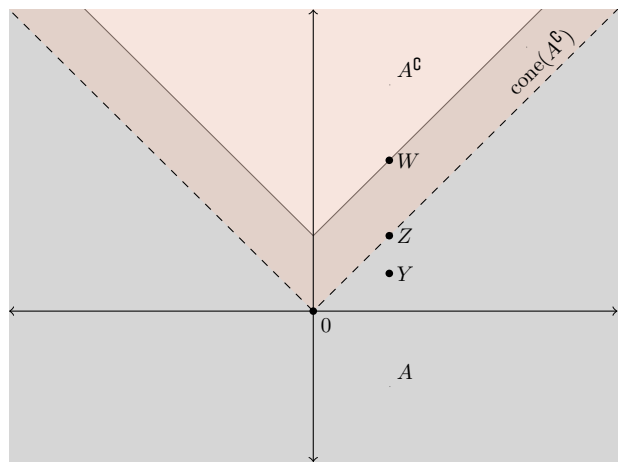
\begin{figure}[H]
\caption{A star-shaped set $A$ (in gray) with convex complement for which $\f$ is not concave.}\label{fig counter concave}
\begin{center}
\begin{tikzpicture}[scale=1]

\clip (0,0) (-4,-2)rectangle (4,4);
\fill[fill=black!20, fill opacity=.8, rotate = 0](0,0) (-4,-4)rectangle (4,4);
\filldraw[ fill=white!20, rotate = 0] (-4,5)--(0,1)-- (4,5);

\filldraw[ fill=BrickRed!20, fill opacity=.50, dashed] (-5,5)--(0,0)-- (5,5);

\fill (1,3) circle [radius=.0cm] node[anchor=south west,scale=.8]{$A^\complement$};
\fill (1,-1) circle [radius=.0cm] node[anchor=south west,scale=.8]{$A$};
\fill (2.8,3.5) circle [radius=.0cm] node[anchor=north,scale=.8, rotate=45]{$\cone (A^\complement)$};
%\draw[dashed, red](-1,-0.5)--(1,1.5);

\fill (0,0) circle [radius=.05cm] node[anchor=north west,scale=.8]{$0$};

\fill (1,1) circle [radius=.05cm] node[anchor=west,scale=.8]{$Z$};
\fill (1,2) circle [radius=.05cm] node[anchor=west,scale=.8]{$W$};
\fill (1,0.5) circle [radius=.05cm] node[anchor=west,scale=.8]{$Y$};

\draw[thin,<->] (-4,0) -- (4,0);
\draw[thin,<->] (0,-2) -- (0, 4);
\end{tikzpicture}
\end{center}
\end{figure}

\end{document}